\documentclass{l4dc2026}
\title[Short Title]{Who Moved My Distribution? Conformal Prediction for Interactive Multi-Agent Systems}
\usepackage{times}
\usepackage{amsmath}
\usepackage{amsfonts}
\usepackage{algorithmic}
\usepackage{textcomp}
\usepackage{graphicx}
\usepackage{subcaption} 
\usepackage{booktabs}
\usepackage{caption}
\usepackage{algorithm}
\usepackage{subcaption}
\usepackage{multirow} 
\usepackage{ragged2e}
\usepackage{mathtools}
\newtheorem{assumption}[theorem]{Assumption}
\newtheorem{problem}{Problem}
\newcommand{\basefirst}{NCP}
\newcommand{\basesec}{BCP}
\newcommand{\scp}{ISCP*}
\newcommand{\cp}{ICP*}

\newcommand{\trajpred}{\texttt{TrajPred}}
\newcommand{\distributionshift}{{endogenous distribution shift}}

\newcommand{\D}{\mathcal{D}}
\newcommand{\T}{\mathcal{T}}

\newcommand{\CP}{\textit{CP}}

\usepackage{xcolor}
\author{%
 \Name{Allen Emmanuel Binny} \Email{allenebinny@kgpian.iitkgp.ac.in}\\
 \addr Department of Electronics and Electrical Communication Engineering, IIT Kharagpur, India
 \AND
 \Name{Anushri Dixit} \Email{anushridixit@ucla.edu}\\
 \addr Department of Mechanical and Aerospace Engineering, University of California, Los Angeles, USA%
}

\begin{document}

\maketitle

\begin{abstract}%
Uncertainty-aware prediction is essential for safe motion planning, especially when using learned models to forecast the behavior of surrounding agents. Conformal prediction is a statistical tool often used to produce uncertainty-aware prediction regions for machine learning models. Most existing frameworks utilizing conformal prediction-based uncertainty predictions assume that the surrounding agents are non-interactive. This is because in closed-loop, as uncertainty-aware agents change their behavior to account for prediction uncertainty, the surrounding agents respond to this change, leading to a distribution shift which we call \textit{endogenous distribution shift}. To address this challenge, we introduce an {iterative conformal prediction framework} that systematically adapts the uncertainty-aware ego-agent controller to the endogenous distribution shift. The proposed method provides probabilistic safety guarantees while adapting to the evolving behavior of reactive, non-ego agents. We establish a model for the endogenous distribution shift and provide the conditions for the iterative conformal prediction pipeline to converge under such a distribution shift. We validate our framework in simulation for 2- and 3- agent interaction scenarios, demonstrating collision avoidance without resulting in overly conservative behavior and an overall improvement in success rates of up to $9.6\%$ compared to other conformal prediction-based baselines.
\end{abstract}


\section{Introduction}
Modern learning-based methods have demonstrated impressive capabilities in predicting multi-agent trajectories. These include architectures such as long short-term memory (LSTM) networks (\cite{Precog, social_lstm, trajectron}), diffusion models (\cite{SICNAV_Diffusion}), and transformers (\cite{motion_transformer, AgentFormer_2021}). For self-driving applications, trajectory predictors are typically trained on data collected from simulators, such as CARLA (\cite{CARLADosovitskiy17}), or large-scale datasets, such as the Waymo Open Dataset (\cite{WaymoOpenDataSet_2024}). However, these predictors cannot provide safety guarantees for collision avoidance and are susceptible to poor performance under distribution shifts when encountering previously unobserved data, potentially leading to catastrophic failures~(\cite{sfkitkatcrash}). Hence, we need uncertainty quantification tools that can provide safety and reliability guarantees for learning-based trajectory predictors. We employ conformal prediction (CP), a statistical tool introduced by~\cite{2005_Vovk_ConformalPrediction} to provide formal guarantees with user-specified probability $1 - \epsilon$ for navigation alongside reactive agents. While previous methods by~\cite{kuipers2024offpolicy, 2023LindemannConformal, pmlr_dixit23a} have introduced CP techniques for safe multi-agent interactions where the policies of the other, non-ego agents are not allowed to change, we uniquely establish a method for multi-agent CP where the non-ego agents can have policies similar to the ego agent and are allowed to respond to changes in the ego agent's behavior.

\begin{figure}[h]
\vspace{-2mm}
    \centering
    \includegraphics[width=1\linewidth]{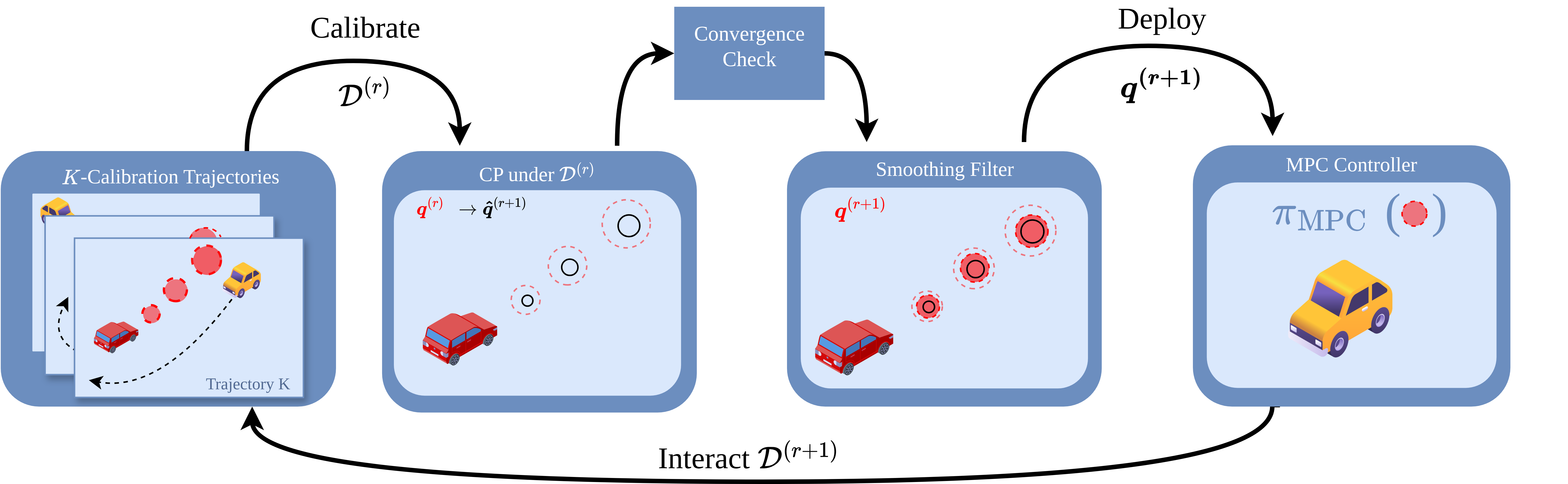}
    \caption{\emph{Iterative Conformal Prediction Framework Overview:} The process alternates between calibration using collected agent trajectories, updating and smoothing the conformal prediction sets, and deploying an uncertainty-aware controller, with iterative adaptation until convergence is reached under \distributionshift.}
    \label{fig:anchor_Figure}
    \vspace{-2mm}
\end{figure}
Incorporating uncertainty predictions into learning-based decision-making in environments with other reactive agents introduces new challenges. Typically, the learned trajectory prediction module is trained with uncertainty-\textit{agnostic} controllers. Consequently, deploying an uncertainty-\textit{aware} controller results in a distribution shift, as other agents adjust their behavior in response to changed ego agent behavior (from an uncertainty-agnostic policy to uncertainty-aware policy). We refer to this phenomenon as \emph{\distributionshift}, which remains largely unaddressed in existing safe motion planning frameworks. Our motivation stems from two sources: 1) the concept of \emph{closed-loop distribution shift} described introduced by~\cite{mei2025perceiveconfidencestatisticalsafety}, where the closed-loop deployment of an uncertainty-quantified perception system combined with a safe planner causes a shift in the state distribution, and 2) the method of \textit{performative prediction} developed by ~\cite{pmlr_perdomo20a} to study the influence of predictive models on the outcomes they aim to predict. Building on these two threads, we provide a multi-agent CP framework that can adapt to endogenous distribution shifts instead of being robust to distribution shifts like~\cite{mei2025perceiveconfidencestatisticalsafety}. The work by~\cite{mei2025perceiveconfidencestatisticalsafety} also provides safety guarantees only for static environments (without other agents) and requires the use of a sampling-based safe planner. Our work addresses performativity for multi-agent decision-making by adapting the  framework by~\cite{pmlr_perdomo20a} to CP-based multi-agent planning. 
% Analogously, we hypothesize that the uncertainty sets employed by other agents create behavioral patterns that are unobserved by the ego agent, thereby inducing the \distributionshift.
% We address the challenges of the \distributionshift~by developing an \textit{iterative uncertainty quantification} method for controllers that utilize learned trajectory predictor outputs.
 The main contributions of this paper are:
\begin{itemize}
    \vspace{-7pt}
    \item We introduce a model for \distributionshift~in uncertainty-aware motion planning, where closed-loop deployment of CP-based predictions alters the data distribution of the interactions between agents and undermines the performance of the controller.
    \vspace{-7pt}
    \item We propose an iterative CP algorithm (as demonstrated in Figure \ref{fig:anchor_Figure}) integrated with model predictive control (MPC) that recalibrates prediction regions using trajectory data collected under the deployed controller, providing probabilistic safety guarantees while accounting for endogenous distribution shifts in reactive multi-agent settings.
    \vspace{-7pt}
    \item We provide theoretical insights for the convergence of our iterative CP algorithm under mild Lipschitzness assumptions.
    \vspace{-7pt}
    \item We demonstrate our results through experiments in a multi-agent environment simulation and compare them to the benchmarks of cases without the iterative framework of CP.
\end{itemize}

\section{Preliminaries}
% \subsection{Trajectory Prediction with LSTMs}
% While our framework is compatible with any trajectory predictor, in practice, we adopt recurrent neural networks (RNNs)~\cite{1997LSTM}, specifically long short-term memory (LSTM) models, due to their ability to capture nonlinear and long-term dependencies.  

% Given observations $(Y_{0},\dots,Y_t)$, an LSTM predictor outputs a sequence of future states $(\hat Y_{t+1|t},\dots,\hat Y_{T|t})$. For each agent $j$, the recurrent update is
% \begin{align}
% h^1_\tau &= H(Y_\tau,h^1_{\tau-1}), \\
% h^i_\tau &= H(Y_\tau,h^i_{\tau-1},h^{i-1}_\tau), \quad i=2,\dots,d, \\
% \hat Y_{\tau+1|\tau} &= \Psi(h^d_\tau),
% \end{align}
% where $h^i_\tau$ is the hidden state of the $i$-th LSTM layer, $H$ is the recurrent update, and $\Psi$ is a decoder mapping the hidden state to the predicted next state.  

% To obtain multi-step predictions, the model is rolled out recursively: the predicted $\hat Y_{\tau+1|\tau}$ is fed back as input to predict $\hat Y_{\tau+2|\tau}$, and so on, up to horizon $T$.

\textbf{Trajectory Prediction:} Given observations $(Y_{0},\dots,Y_t)$ at time $t$, we want to predict future states $(Y_{t+1},\ldots,Y_{t+H})$ for a prediction horizon of $H$. Assume that \trajpred~is a function that maps observations $(Y_{0},\dots,Y_t)$ to predictions $(\hat{Y}_{t+1|t},\ldots,\hat{Y}_{t+H|t})$. Note that $t$ in $\hat{Y}_{t+\tau|t}$ denotes the time at which the prediction is made, while $\tau$ indicates how many steps we predict ahead. To obtain multi-step predictions, the model is rolled out recursively: the predicted $\hat Y_{t+1|t}$ is fed back as input to predict $\hat Y_{t+2|t}$, and so on, up to horizon $H$. While our framework is compatible with any trajectory predictor, in practice, we adopt recurrent neural networks (RNNs) (\cite{1997LSTM}), specifically long short-term memory (LSTM) models, due to their ability to capture nonlinear and long-term dependencies.
% For the CARLA experiments, we perform the training on specific well-trained models on large scale data and show our results with our framework on these well-trained networks.

% \subsection{Conformal Prediction}

\textbf{Conformal Prediction:} Conformal prediction (CP) (\cite{2005_Vovk_ConformalPrediction, angelopoulos2025conformalriskcontrol}) will be our primary tool for performing rigorous uncertainty quantification for trajectory prediction. Given $K$ i.i.d. (or exchangeable) samples $U_1, \ldots, U_K$ of a scalar random variable $U$, we compute the threshold, $\hat{q}_{1-\epsilon}$, such that the next sample, $U_{\text{test}}$, satisfies,
\begin{equation}\label{eq: CP_guarantee}
\mathbb{P}[U_{\text{test}} \leq \hat{q}_{1-\epsilon}] \geq 1 - \epsilon,
\end{equation}
% \begin{equation}\label{eq: CP_threshold}
% \hat{q}_{1-\epsilon} = \begin{cases}
% U_{(\lceil (K+1)(1-\epsilon) \rceil)} & \text{if } \lceil (K+1)(1-\epsilon) \rceil \leq K, \\
% \infty & \text{otherwise},
% \end{cases}
% \end{equation}
\begin{align}\label{eq: CP_threshold}
\hat{q}_{1-\epsilon} &= \begin{cases}
U_{(\lceil (K+1)(1-\epsilon) \rceil)} & \text{if } \lceil (K+1)(1-\epsilon) \rceil \leq K, \\
\infty & \text{otherwise},
\end{cases}
\end{align}
and $U_{(1)} \leq U_{(2)} \leq \ldots \leq U_{(K)}$ are the order statistics (sorted values) of the $K$ samples $U_1, \ldots, U_K$. In the CP literature, $U$ is known as the non-conformity score and it is a measure of the (in)correctness of a model. The above guarantee is \emph{marginal}, i.e., it holds over the sampling of both the calibration dataset $U_1, \ldots, U_K$ and the test variable $U_{\text{test}}$. Hence, we need to generate a fresh set of i.i.d. calibration data $U_1, \ldots, U_K$ for the guarantee to hold for a new sample $U_{\text{test}}$. However, in practice, one typically only has access to a single dataset of examples; inferences from this dataset must be used for all future predictions on test examples. In this work, we use the following dataset-conditional guarantee (\cite{pmlr-v25-vovk12,angelopoulos2022gentleintroductionconformalprediction}) that does not require us to generate $K$ new samples for every test prediction and holds with probability $1 - \delta$ over the sampling of the calibration dataset,
\begin{equation}\label{eq:CP_Prediction}
    \mathbb{P}[U_{\text{test}} \leq \hat{q}_{1-\epsilon}|U_1, \ldots, U_K] \geq \text{Beta}_{K+1-v,v}^{-1}(\delta),
\end{equation}
where, $v := \lfloor (K + 1)\epsilon \rfloor$, and $\text{Beta}_{K+1-v,v}^{-1}(\delta)$ is the $\delta$-quantile of the Beta distribution with parameters $K + 1 - v$ and $v$. We can choose $\epsilon$ to achieve the desired $1 - \epsilon$ coverage.

\textbf{Notation:} The system operates with $N$ other dynamic agents whose trajectories are apriori unknown. We define an iterate $r$, for which $\D^{(r)}$ denotes an unknown distribution over agent trajectories collected at that iteration. We sample an interaction between agents by sampling a random trajectory $Y^{(r)}\coloneqq (Y_{0}^{(r)},Y_{1}^{(r)},\dots) \sim \mathcal{D}^{(r)}$ where the joint agent state at time $t$ is given by $Y_{t}^{(r)}\coloneqq(Y_{t,0}^{(r)},\ldots,Y_{t,N}^{(r)})\in \mathbb{R}^{(N+1)\times n}$,\footnote{We assume the state of each agent to be $n$-dimensional. This can easily be generalized.} {i.e.}, \(Y_{t,j}\) is the state of agent $j$ at time $t$ (agent $0$ is used to denote the ego agent). During the calibration phase, for each iteration $r$, we form a dataset of $K$ interaction trajectories, $D_{cal}^{(r)}\coloneqq\{Y^{(r),1},\ldots,Y^{(r),K}\}$ in which each of the $K$ trajectories are exchangeable and $Y^{(r),i}\coloneqq\{Y_0^{(r),i},Y_1^{(r),i},\ldots\}\sim\mathcal{D}^{(r)}$. For simplicity, sometimes we drop the notation for each iteration $r$ when the same algorithm is applied for all the iterations in the same way, {i.e.}, $Y^{(r)}$ is sometimes denoted as $Y$. Similarly, for the rest of the paper, we consider the two agent case, i.e., $N=1$, but our results hold for any $N>1$. The same notation is used for predicted trajectories, but instead of $Y$, we denote predicted trajectories by $\hat{Y}$. We denote, $\mathcal{D}^* := \lim_{r\rightarrow \infty}\mathcal{D}^{(r)}$ and $Y^* := \lim_{r\rightarrow \infty}Y^{(r)}$.

\section{Problem Statement}
\subsection{System Dynamics and Environment}
We consider an (ego) agent governed by the discrete-time dynamics,
\begin{equation}
\label{eq:dyn}
x_{t+1} = f(x_t,u_t), \qquad x_0 = x(0),
\end{equation}
with state $x_t \in \mathcal{X} \subseteq \mathbb{R}^n$ and control input $u_t \in \mathcal{U} \subseteq \mathbb{R}^m$ at time $t$. 
The sets $\mathcal{X}$ and $\mathcal{U}$ denote state and input constraints, respectively, $f:\mathcal{X}\times\mathcal{U}\to \mathcal{X}$ describes the dynamics, and $x(0) \in \mathbb{R}^n$ is the initial condition.

The ego agent navigates an environment with $N$ other dynamic agents (we present the $N=1$ case for simplicity) that can plan and react to the behavior of the other agents in a decentralized manner. Within its motion planning framework, each agent predicts other agents' locations using any learned trajectory predictor module, \trajpred~ (\cite{trajectron, social_lstm}). In this work, we aim to make uncertainty-aware predictions on the behavior of such reactive agents so that we can provide statistical assurances for the safety of the system~\eqref{eq:dyn} without any information about the dynamics and control of the dynamic agents in the environment. We utilize a model predictive control (MPC) strategy with cost $J$ and prediction horizon $H$ to ensure safety, 
\begin{subequations}\label{eq:open_loop}
\vspace{-2mm}
\begin{align}
    \pi_{\text{MPC}}(\boldsymbol{d}_{\text{min}}) = \min_{(u_{t|t},\hdots,u_{t+H-1|t})}& \sum_{\kappa=0}^{H-1}J(x_{t+\kappa+1|t},u_{t+\kappa|t}) \\
     \text{s.t.}\qquad & x_{t+\kappa|t}=f(x_{t+\kappa-1|t},u_{t+\kappa-1|t}), \\
     & c(x_{t+\kappa|t},\hat{Y}_{t+\kappa|t})\ge d_{\text{min}, \kappa},\\
     & u_{t+\kappa|t} \in \mathcal{U},x_{t+\kappa+1|t} \in \mathcal{X},&\\
      & x_{t|t} = x(t),&
\end{align}
\end{subequations}
where, $x_{t+\kappa|t}, u_{t+\kappa|t}$ describe the state and control predicted for time ${t+\kappa}$ at time $t$ for $\kappa\in\{1,\hdots, H\}$, and $\hat{Y}_{t+\kappa|t}$ describes the $H$-step predicted trajectories of each of the $N$ other agents at time $t$ using \trajpred. The collision avoidance safety constraint is described by the function $c$ that measures the predicted distance between the ego agent state $x_{t+\kappa|t}$ and that of the other agents in the scene, i.e., $c(x_{t+\kappa|t},\hat{Y}_{t+\kappa|t}) = \lVert x_{t+\kappa|t} - \hat{Y}_{t+\kappa|t} \rVert_2$. 

Based on the accuracy of \trajpred, the predicted states of the other agents, $\hat{Y}_{t+\kappa|t}$, may be imperfect. To ensure the safety of the ego agent, we quantify the uncertainty associated with \trajpred~using an uncertainty-aware collision threshold, $d_{\text{min}, \kappa} = \hat{q}_\kappa$ using CP. We assume that we have access to a calibration dataset of i.i.d. interactions between all agents. We compute the prediction error $U =\lVert Y - \hat{Y}\rVert$ using the actual realized trajectory of the agents, $Y$ and the predicted trajectory, $\hat{Y}$, to estimate a statistically valid $\boldsymbol{q} = [q_1,\ldots,q_H]^T$. For simplicity, we consider that the other $N$ agents utilize an MPC planner \(\pi_{\text{MPC}}(\boldsymbol{0})\). However, this can be generalized to another planner with obstacle avoidance, such as using control barrier functions (\cite{2017CBFAmes}) or dynamic window approach (\cite{1997FoxDWA}).

\subsection{Endogenous Distribution Shifts}
When collecting the initial calibration dataset, we utilize a nominal controller without any uncertainty-aware constraint tightening, i.e., $\pi_{\text{MPC}}(\boldsymbol{0})$ with $\boldsymbol{d}_{\text{min}}=\boldsymbol{0}$. We then estimate the desired tightening $\boldsymbol{d}_{\text{min}} = \boldsymbol{q}$ using CP (defined in section \ref{Predictive Region Construction}) leading to a new, uncertainty-aware policy, $\pi_{\text{MPC}}(\boldsymbol{q})$ that is reliable and safe for any environment drawn from the distribution $\mathcal{D}^{(0)}$ with  the desired probability. Since other agents in the environment have reactive policies, their behaviors will now change in response to the new policy leading to an unknown distribution shift in the interaction trajectories $\D^{(0)} \rightarrow \D$. We define this distribution shift caused by internal changes to the system controller \emph{\distributionshift}. Our goal in this work is to ensure the safety of the ego agent in such a dynamic and shifting environment without comprising on the performance of the controller.

\begin{problem}
Consider the discrete-time dynamical system given in~\eqref{eq:dyn} with a control law governed by $u = \pi_{\text{MPC}}(\boldsymbol{0})$ that accounts for the motion of other unknown agents in the environment, $Y^{(0)}_{t}\sim \mathcal{D}^{(0)}$ using a learned trajectory predictor, \trajpred. Compute a constraint tightening, $\boldsymbol{q}^*$ associated with a stable, non-shifting distribution $\mathcal{D}^*$, such that $\pi_{\text{MPC}}(\boldsymbol{q}^*)$ is safe with probability $1-\epsilon$ with,
\begin{equation}
\mathbb{P}\left(\|Y^*_{t+\kappa} - \hat Y_{t+\kappa|t}\| \leq q_\kappa^*\right) \geq 1-\epsilon, \quad \forall \kappa \in \{1,\dots,H\},
\end{equation}
where, the realized agent interactions under the control law $\pi_{\text{MPC}}(\boldsymbol{q}^*)$ satisfy $Y^*_{t+\kappa}\sim \mathcal{D}^*$ and $Y^*_{t+\kappa} - \hat Y_{t+\kappa|t}$ is the difference between the actual trajectories of all agents at time $t+\kappa$ and the predicted trajectories as predicted by \trajpred~at time $t$ with failure probability $\epsilon\in[0,1)$. 
\end{problem}

\section{Conformal Prediction under Endogenous Distribution Shift}
We now describe our {iterative conformal prediction} framework to adapt the CP-based uncertainty predictions to any \distributionshift. As introduced earlier, an \distributionshift~occurs when the environment, comprising reactive agents, changes in response to a shift in the ego agent's controller. We begin by formalizing a mathematical model for this distribution shift.
\subsection{Modeling Endogenous Distribution Shift}\label{endodistributionshift}
Let $\mathcal{Y}$ be the space of all possible trajectories and $\mathcal{P}(\mathcal{Y})$ be the set of all probability distributions on $\mathcal{Y}$. We can describe a metric space $(\mathcal{Y}, W_1)$, where $W_1$ is the $1-$Wasserstein distance.\footnote{We may alternatively equip the metric space with another metric like the Total Variation Distance or the Kolmogorov–Smirnov (KS) metric.} We propose to capture the \distributionshift ~through an iterative map $\mathcal{T}:\mathcal{P}(\mathcal{Y}) \to \mathcal{P}(\mathcal{Y})$ defined as
\begin{equation}
    \D^{(r+1)} = \mathcal{T}(\D^{(r)}), \quad \text{where} \quad \mathcal{T} = F \circ \CP.
\end{equation}
Here, ${\D^{(r)}}\in \mathcal{P}(\mathcal{Y})$ is the distribution of interaction trajectories calibrated at iteration $r$, and $\boldsymbol{q}^{(r)} = \CP(\D^{(r)})$ is the uncertainty threshold obtained from CP applied to $\D^{(r)}$ (as detailed in Section \ref{Predictive Region Construction}).  Effectively, the conformal prediction operator $\CP$ maps the trajectory distribution $\D^{(r)}$ to a statistically valid prediction error $\boldsymbol{q}^{(r)}$, using a calibration dataset $D_{cal}^{(r)} \coloneqq \{Y^{(r),1},\ldots,Y^{(r),K}\}$, where each trajectory $Y^{(r),i} \sim \D^{(r)}$. The mapping $F$ is the new trajectory distribution induced by running $\pi_{\text{MPC}}(\boldsymbol{q}^{(r)})$, as seen in Equation \ref{eq:open_loop}. Thus, $\mathcal{T}$ models the combined effect of uncertainty quantification through CP and the resulting trajectory adaptations by the MPC planner. This formulation enables us to capture how the behaviors of other agents evolve in response to changes in the MPC controller of the ego agent.
Our objective is to find a target controller $\pi^*_{\text{MPC}}$ that induces no further distribution shift, {i.e.}, \(\D^{*} = \mathcal{T}(\D^*)\).
% \begin{equation}
%     \D^{*} = \mathcal{T}(\D^*)
%     % \quad \text{and} \quad \Pr(\pi^*_{\text{MPC}} \text{ is safe} \ge 1-\epsilon).
% \end{equation}
% We provide the theoretical convergence guarantees for the iterative conformal prediction framework by proving that the combined operator is a contraction mapping under the Wasserstein metric. 
% Consider an iterative loop where at iteration $r$, the system has trajectory distribution $\mathcal{D}^{(r)}$. The conformal prediction set $q^{(r)} = \text{CP}(D^{(r)})$ is computed from $D^{(r)}$, and the MPC planner $F$ produces a new distribution $D^{(r+1)} = F(q^{(r)})$. This defines the iterative map:
% \begin{equation}
% D^{(r+1)} = \mathcal{T}(D^{(r)}), \quad \text{where} \quad \mathcal{T} = F \circ \text{CP}
% \end{equation}
In other words, if $\T$ has a fixed point and show convergence of our end-to-end prediction and planning system to a unique fixed-point distribution $\mathcal{D}^*$, then we can find this stable solution and controller using an iterative algorithm. 

\begin{assumption}[Lipschitz Continuity of Conformal Prediction]\label{ass:cp_lipschitz}
The conformal prediction mapping $\text{CP}$ is Lipschitz continuous w.r.t. the Wasserstein distance. There exists $L_{\text{CP}} \geq 0$ such that for any two trajectory distributions $\mathcal{D}_1$ and $\mathcal{D}_2$,
\begin{equation}\label{eq:lip_cp}
|\text{CP}(\mathcal{D}_1) - \text{CP}(\mathcal{D}_2)| \leq L_{\text{CP}} \, W_1(\mathcal{D}_1, \mathcal{D}_2).
\end{equation}
\end{assumption}
\begin{remark}[Justification of Assumption~\ref{ass:cp_lipschitz}]
The above assumption relates to the shift of the quantile estimation process under changes in the underlying data distribution. The work~\cite[Proposition 2.5]{LP_CP_2025} shows that if the score function $U:\mathcal{Y}\rightarrow\mathbb{R}$ is $k$-Lipschitz, then a data-space distribution shift translates to a corresponding shift in the distribution of the nonconformity scores.
% They model these shifts using LP(Lévy-Prokhorov) ambiguity sets. 
In~\cite[Proposition 3.5]{LP_CP_2025}, the authors provide an expression for the worst-case quantile and coverage using the parameters \(\epsilon\) and $\rho$ of the distribution shift. Thus, under mild assumptions, we can show that $L_{\text{CP}} \leq1$. 
% Furthermore, the smoothing mechanism described in Section~\ref{sec:smooth} regularizes the distribution updates, ensuring continuity between successive iterations and enhancing the stability of the convergence analysis.
\end{remark}

% The closed-loop deployment of MPC planner \(\pi_{\text{MPC}}(.)\) requires the CP-based uncertainty predictions $\boldsymbol{q}^{(r)}$ as its input and induces a new distribution $\mathcal{D}^{(r+1)}$. We define this mapping as $F$.
\begin{assumption}[Lipschitz Continuity of MPC]
\label{ass:mpc_lipschitz}
The MPC map $F$ is Lipschitz continuous w.r.t. changes in the constraint set. There exists $L_{\text{MPC}} \geq 0$ such that for any two constraint sets $\boldsymbol{q}_1$ and $\boldsymbol{q}_2$,
\begin{equation}\label{eq: lip_mpc}
W_1(F(\boldsymbol{q}_1), F(\boldsymbol{q}_2)) \leq L_{\text{MPC}} \, |\boldsymbol{q}_1 - \boldsymbol{q}_2|.
\end{equation}
\end{assumption}
\begin{remark}[Justification of Assumption~\ref{ass:mpc_lipschitz}]
This assumption is about the sensitivity of the MPC's solution (the resulting trajectory distribution) to changes in the constraints. If the MPC cost function is strongly convex and the constraints vary smoothly (without abrupt switching), the optimal solution will change continuously with the parameters (locally)~(\cite{boot1963sensitivity, canovas2025lipschitzCO}).
% The smoothing filter ensures the update is gradual for the change of conditions so that they do not vary rapidly.
\end{remark}
\begin{theorem}[Convergence to the fixed-point of the \distributionshift~map]
\label{thm:convergence}
If Assumptions~\ref{ass:cp_lipschitz} and~\ref{ass:mpc_lipschitz} hold and the combined Lipschitz constant satisfies $\delta = L_{\text{MPC}} \cdot L_{\text{CP}} < 1$, then the operator $\mathcal{T} = F \circ \CP$ is a contraction mapping on the space of trajectory distributions equipped with the Wasserstein metric, $(\mathcal{Y}, W_1)$. Consequently, the iteration $\mathcal{D}^{(r+1)} = \mathcal{T}(\mathcal{D}^{(r)})$ converges to a unique fixed-point distribution $\mathcal{D}^* = \lim_{r\rightarrow \infty}\mathcal{D}^{(r)}$.
\end{theorem}
\begin{proof}
    When we combine~\eqref{eq:lip_cp} and ~\eqref{eq: lip_mpc}, we get, 
    \begin{equation*}
        W_1(\T(\mathcal{D}_1), \T(\mathcal{D}_2))= W_1\Big(F\big(CP(\mathcal{D}_1)\big), F\big(CP(\mathcal{D}_2)\big)\Big) \leq L_{\text{MPC}}. L_{\text{CP}} \, W_1(\mathcal{D}_1, \mathcal{D}_2).
    \end{equation*}
   If $\delta=L_{\text{MPC}}.L_{\text{CP}}<1$, we can guarantee convergence to a fixed-point by leveraging Banach's Fixed-Point Theorem~(\cite{banach1922operations}). Hence, $\mathcal{T}$ is a contraction mapping with $\mathcal{D}^* = \lim_{r\rightarrow \infty}\mathcal{D}^{(r)}$.
\end{proof}

Since $\D^{*} = \mathcal{T}(\D^{*})$, convergence of the distributions implies convergence of CP-based uncertainty predictions $\boldsymbol{q}^{(r)}$ to $\boldsymbol{q}^*$. This statistic can thus serve as a stopping criterion in practice. Thus, we now require the formulation for \(\boldsymbol{q}^{(r)}\).
\subsection{Prediction Region Construction}\label{Predictive Region Construction}
We now describe how to construct a CP-based prediction to be utilized in the MPC planner~\eqref{eq:open_loop}. Given past observations, $(Y_{0}^{(r)},\dots,Y_t^{(r)})$, \trajpred ~generates an $H$-step forecast, $(\hat{Y}_{t+1|t}^{(r)}, \dots, \hat{Y}_{t+H|t}^{(r)})$. During the calibration phase, for each iteration $r$, we form a dataset of $K$ interaction trajectories, $D_{cal}^{(r)}\coloneqq\{Y^{(r),1},\ldots,Y^{(r),K}\}$ and $Y^{(r),i}\coloneqq\{Y_0^{(r),i},Y_1^{(r),i},\ldots\}\overset{\mathrm{i.i.d.}}{\sim}
\mathcal{D}^{(r-1)}$.
For each $i^{\text{th}}$ calibration trajectory $Y^{(r),i} \in D_{\text{cal}}$, and each horizon step $\kappa$, we define the nonconformity score,
\vspace{-2mm}
\begin{equation}
    U^{(r),i}_{\kappa|t} = \big\| Y^{(r),i}_{t+\kappa} - \hat{Y}^{(r),i
    }_{t+\kappa|t} \big\|_2 \quad \forall \kappa \in \{1,\ldots,H\},~i\in\{1,\ldots,K\},
    \vspace{-2mm}
\end{equation}
which measures the deviation between the predicted and ground truth state. 
We define the maximal nonconformity score across time and across all $N$ agents,
\vspace{-2mm}
\begin{equation}
    U^{(r),i}_\kappa = \max_{t \in \{1, 2,\dots\}} 
   U_{\kappa|t}^{(r),i} \quad \forall \kappa \in\{1,\ldots,H\},~i\in\{1,\ldots,K\}.
   \vspace{-2mm}
\end{equation}
The nonconformity scores that describe the maximal $H$-step prediction error across the entire trajectory, $\{U^{(r),i}\}_{i=1}^K$ are exchangeable. Assuming they are sorted in non-decreasing order, we can apply~\eqref{eq: CP_threshold}  at horizon $\kappa$, to obtain the threshold,
\vspace{-2mm}
\begin{equation}
    \hat{q}^{(r)}_\kappa = U^{(r),p}_\kappa, 
    \qquad 
    p = \left\lceil (1-\epsilon)(K+1) \right\rceil,
    \vspace{-2mm}
\end{equation}
where, $\epsilon \in (0,1)$ is the failure probability. Hence we obtain \(\boldsymbol{\hat{q}}^{(r)}\coloneqq[\hat{q}_1^{(r)},\ldots,\hat{q}_H^{(r)}]\) as the CP-based threshold such that it satisfies~\eqref{eq: CP_guarantee} in a new interaction, i.e., 
\vspace{-2mm}
\begin{equation*}
    \mathbb{P}\bigg[ \max_{t \in \{1, 2,\dots\}}\big\| Y^{(r),\text{test}}_{t+\kappa} - \hat{Y}^{(r),\text{test}}_{t+\kappa|t} \big\|_2 \leq \hat{q}^{(r)}_\kappa\bigg] \geq 1 - \epsilon, \quad \forall \kappa \in\{1,\ldots,H\},\,\, Y^{(r),\text{test}}\overset{\mathrm{i.i.d.}}{\sim}\mathcal{D}^{(r-1)}.
    \vspace{-2mm}
\end{equation*}
We can utilize the above CP-based prediction set in the safety constraint of our MPC problem, i.e., $\pi_{\text{MPC}}(\boldsymbol{\hat{q}}^{(r)})$ which will ensure that the ego-agent plans safely using imperfect predictions with probability $1-\epsilon$~\cite[Theorem 3]{2023LindemannConformal}.
% Since the number of time steps may vary in each trajectory at each iteration, we consider this generalized framework, where the prediction set $\boldsymbol{q}^{(r)}$ is defined over the time horizon $H$. 
Our method is compatible with any CP-based frameworks such as those by~\cite{Cauchois_2024, cleaveland2024conformalpredictionregionstime}.

% \begin{theorem}[Validity of max-horizon predictive regions]
% \label{thm:max_region}
% Let $U^{(r,i)}_\kappa$ be the max-horizon nonconformity scores defined above and let 
% $q^{(r)}_\kappa$ be the prediction set as the $(1-\epsilon)$ empirical quantile 
% of $\{U^{(r,i)}_\kappa\}_{i=1}^K$. Then for a test trajectory $Y^{(0)}$ drawn from an exchangeable calibration set,
% \(\Pr\!\left(
%         U_\kappa^{(r,i)}
%     \right) \;\geq\; 1-\epsilon\). 
% Equivalently, with probability at least $1-\epsilon$,
% \begin{equation}
%     \big\| Y^{(0)}_{t+\kappa} - \hat{Y}^{(0)}_{t+\kappa|t} \big\|
%     \;\leq\; q^{(r)}, \qquad \forall \kappa \in \{1,\dots,H\}.
% \end{equation}
% \end{theorem}

% \begin{proof}
% The exchangeability of $\{U^{(r,i)}_\kappa\}_{i=1}^K$ with the test score 
% $R^{\text{0}}$ implies that the conformal quantile $q^{(r)}$ satisfies
% \begin{equation*}
% \begin{aligned}
% \Pr\!\Big(
%   \max_{t\in\mathcal{T}}
%   \big\|Y^{(0)}_{t+\kappa}-\hat Y^{(0)}_{t+\kappa\mid t}\big\|
%   \le q^{(r)}_{\kappa}
% \Big) \ge 1-\epsilon,
% \quad &\forall \kappa\in\{1,\dots,H\},\\
% \Pr\!\Big(
%   \big\|Y^{(0)}_{t+\kappa}-\hat Y^{(0)}_{t+\kappa\mid t}\big\|
%   \le q^{(r)}_{\kappa}
% \Big) \ge 1-\epsilon,
% \quad &\forall \kappa\in\{1,\dots,H\} \\
% &\forall t
% \end{aligned}
% \end{equation*}
% Thus conformal prediction holds for the test sample drawn from the same distribution as the calibration dataset.
% \end{proof}

\subsection{Smoothing Filter for Prediction Set}\label{sec:smooth}
As established in Section \ref{endodistributionshift}, at the fixed-point $\mathcal{D}^*$, we have a corresponding calibration threshold from $\CP$ given by $\boldsymbol{q}^*$ and a stable, fixed policy, $\pi_{\text{MPC}}(\boldsymbol{q^*})$. To ensure stable convergence to this fixed-point, we constrain the distribution shift over the iterations by applying a low-pass filter to the $\boldsymbol{\hat{q}}^{(r)}$ update. This reduces the short-term fluctuations in the distribution shift across iterations.

% the convergence of the CP uncertainty sets $\boldsymbol{q}^{(r)}$ corresponds to the convergence of the trajectory distribution under the MPC controller to the fixed point $\mathcal{D}^*$. To ensure this stable convergence, the updates in the MPC planner $\pi_{\text{MPC}}(\boldsymbol{q}^{(r)})$ should be smooth, which can be ensured by smoothing the parameter updates $\boldsymbol{q}^{(r)}$ themselves. Additionally,~\cite {pmlr_perdomo20a} assumes that the distribution for each consecutive iteration is $\epsilon\text{-bounded}$ for the stability of the iterative procedure.

We apply a smoothing filter on the prediction set update by filtering the calibration threshold. Given a new calibration threshold computed at iteration $r$, $\boldsymbol{\hat{q}}^{(r+1)}$, using calibration data sampled from $\mathcal{D}^{(r)}$, the filtered threshold is given by,
\vspace{-2mm}
\begin{equation}
    \boldsymbol{q}^{(r+1)} \;=\; (1-\gamma)\,\boldsymbol{q}^{(r)}
    \;+\; \gamma\,\boldsymbol{\hat{q}}^{(r+1)},
    \vspace{-2mm}
    \label{eq:radius_update}
\end{equation}
where, $\gamma\in[0,1]$ is a hyperparameter to determine the extent of smoothness. This iterative smoothing serves to ensure that the induced trajectory distributions $\mathcal{D}^{(r)}$, corresponding to executing $\pi_{\text{MPC}}(\boldsymbol{q}^{(r)})$, remain close to the preceding distributions $\mathcal{D}^{(r-1)}$. Thus, we constrain the system to evolve gradually and facilitate convergence to the fixed-point.

Algorithm~\ref{alg:iterative_cp} summarizes our iterative adaptation framework. Theorem \ref{thm:convergence} guarantees asymptotic convergence of Algorithm~\ref{alg:iterative_cp} to a fixed-point where, $\lim_{r\rightarrow\infty}\lVert\boldsymbol{q}^{(r+1)} - \boldsymbol{q}^{(r)}\rVert_2 = 0$. However, practical deployment in requires a stopping criterion to detect when steady state is approximately reached. We therefore consider the following stopping condition with a small constant $\phi$\footnote{For simplicity, we have considered the interaction between two agents, i.e., $N=1$, but note that we can apply the same updates for all non-ego agent ($N>1$) individually. We update the stopping criteria to \(\max_{j\in\{0,\ldots,N\}} \big\| \boldsymbol{q}^{(r+1)}_j - \boldsymbol{q}^{(r)}_j \big\| \leq \phi\).},
\vspace{-2mm}
\begin{equation}\label{eqn:convergence}
    \big\| \boldsymbol{q}^{(r+1)} - \boldsymbol{q}^{(r)} \big\|_2 \leq \phi.
    \vspace{-2mm}
\end{equation}
Hence, our framework described in Algorithm~\ref{alg:iterative_cp}, is able to converge to a fixed-point, under the assumptions listed in Theorem~\ref{thm:convergence}, and guarantee probabilistic safety (with probability $1-\epsilon$) of the MPC policy $\pi_{\text{MPC}}(\boldsymbol{q}^*)$ in the presence of other reactive agents.

\begin{algorithm}[t]
% \footnotesize
\caption{Iterative Conformal Prediction for Multi-Agent Interactive Systems}
\label{alg:iterative_cp}
\begin{algorithmic}[1]
% \REQUIRE Trajectory Predictor (\trajpred), failure probability ($\epsilon$), smoothing ($\gamma$), stopping criteria ($\phi$), No. of agents ($N$)
\REQUIRE Failure probability $\epsilon$, smoothing parameter $\gamma$, stopping criteria $\phi$, prediction horizon $H$
\STATE Initialize $r \leftarrow 0$; $\boldsymbol{q}^{(0)} \leftarrow \boldsymbol{0}$; $p \leftarrow \lceil (K+1)(1-\epsilon)\rceil$

\REPEAT
    \STATE Deploy \(\pi_{\text{MPC}}(\boldsymbol{q}^{(r)})\) with \trajpred ~to obtain predictions
    \STATE Collect trajectories for $D_{cal}^{(r)}$ 
    \FOR{$\kappa$ from $1$ to $H$}
        \STATE 
        
        $U^{(r,i)}_{\kappa} \leftarrow \max_t \big\|Y^{(r),i}_{t+\kappa} - \hat{Y}^{(r),i}_{t+\kappa|t}\big\|_2
        $ for each $Y^{(r),i} \in D_{cal}^{(r)}$
        % \STATE $U^{(r),|D_{cal}|+1}_{\kappa} \leftarrow \infty$
        % \STATE Sort $U^{(r),i}_{\kappa}$ in non-decreasing order
        \STATE Apply~\eqref{eq: CP_threshold} where $\hat{q}^{(r+1)}_{\kappa} = U^{(r),p}_{\kappa}$
    \ENDFOR

    \STATE $\boldsymbol{q}^{(r+1)} \leftarrow (1-\gamma). \boldsymbol{q}^{(r)} + \gamma . \boldsymbol{\hat{q}}^{(r+1)}$
    \STATE $\Delta q = \| \boldsymbol{\hat{q}}^{(r+1)} - \boldsymbol{q}^{(r)} \|_2$
    \STATE $r \leftarrow r+1$
\UNTIL{$\Delta q \leq \phi$}

\RETURN $\boldsymbol{q}^{(r)}$
\end{algorithmic}
\end{algorithm}

\vspace{-5mm}
\section{Results}
% In the synthetic environment, our iterative conformal prediction (CP) framework addresses distributional shifts effectively, achieving the desired statistical guarantees without excessive conservatism. This leads to tighter uncertainty bounds, reduced deviations from reference trajectories, and higher success rates. Environments with two or three dynamic agents are considered, where each agent is randomly initialized within a circle and assigned an independent goal per trial.

In this section, we apply our iterative CP framework in two- and three-agent interactions. We consider that all agents can react to each other using the policy $\pi_{\text{MPC}}$ where each agent makes trajectory predictions about the other agents based on past observations. The \trajpred~method considered is a long short-term memory (LSTM) model trained on historical agent states, chosen for scalability with time-series data and limited training samples. Agents implement model predictive control (MPC) for collision avoidance, with the ego agent integrating CP uncertainty regions around other agents. Further details regarding the experimental setup is included in Appendix \ref{appendix:env_setup}. We compare four approaches:

\vspace{-5pt}
\setlength{\leftmargini}{15pt} % remove left margin for first-level lists
\begin{enumerate}
    \item \textbf{No Conformal Prediction (NCP):}\label{NCP}  No uncertainty quantification is applied, \emph{i.e.}, $\pi_{\text{MPC}}(\boldsymbol{0})$.
    \vspace{-7pt}
    \item \textbf{Baseline Conformal Prediction (BCP):}\label{BCP} Trajectories are collected from decentralized multi-agent control with interactive agents. A trajectory predictor is fine-tuned on these trajectories (using $D_{tune}^{\text{BCP}}$ trajectories) , and a calibration dataset \(D_{cal}^{\text{BCP}}\) is constructed for uncertainty calibration once, i.e., $r=1$, similar to the method proposed by~\cite{2023LindemannConformal}.
    \vspace{-7pt}
    \item \textbf{Iterative Conformal Prediction (ICP, Ours):}\label{ICP} Utilizing the same fine-tuned predictor as BCP, the iterative CP framework is executed until convergence~\eqref{eqn:convergence}. For fairness, the total number of trajectories matches that of BCP but is distributed across iterations, i.e., $|D_{cal}^{\text{BCP}}| = r.K$.
    \vspace{-7pt}
    \item \textbf{Iterative Split Conformal Prediction (ISCP, Ours):}\label{ISCP} The Iterative Split Conformal Prediction framework is based on split-conformal prediction as shown in ~\cite{Tibshirani2023Conformal} where the predictor is retrained and prediction sets are calibrated iteratively. The predictor is trained with a dataset size as the predictor trained in BCP, i.e, $|D_{tune}^{\text{BCP}}| = r.K_{tune}$ and $|D_{cal}^{\text{BCP}}| = r.K$. This framework is applied until the convergence criteria~\eqref{eqn:convergence} is satisfied.
\end{enumerate}
\vspace{-4mm}
\subsection{Two Agent Simulations}\label{2_agent_results}
\vspace{-2mm}
\begin{figure}[h]
    \centering
    \includegraphics[width=1\linewidth]{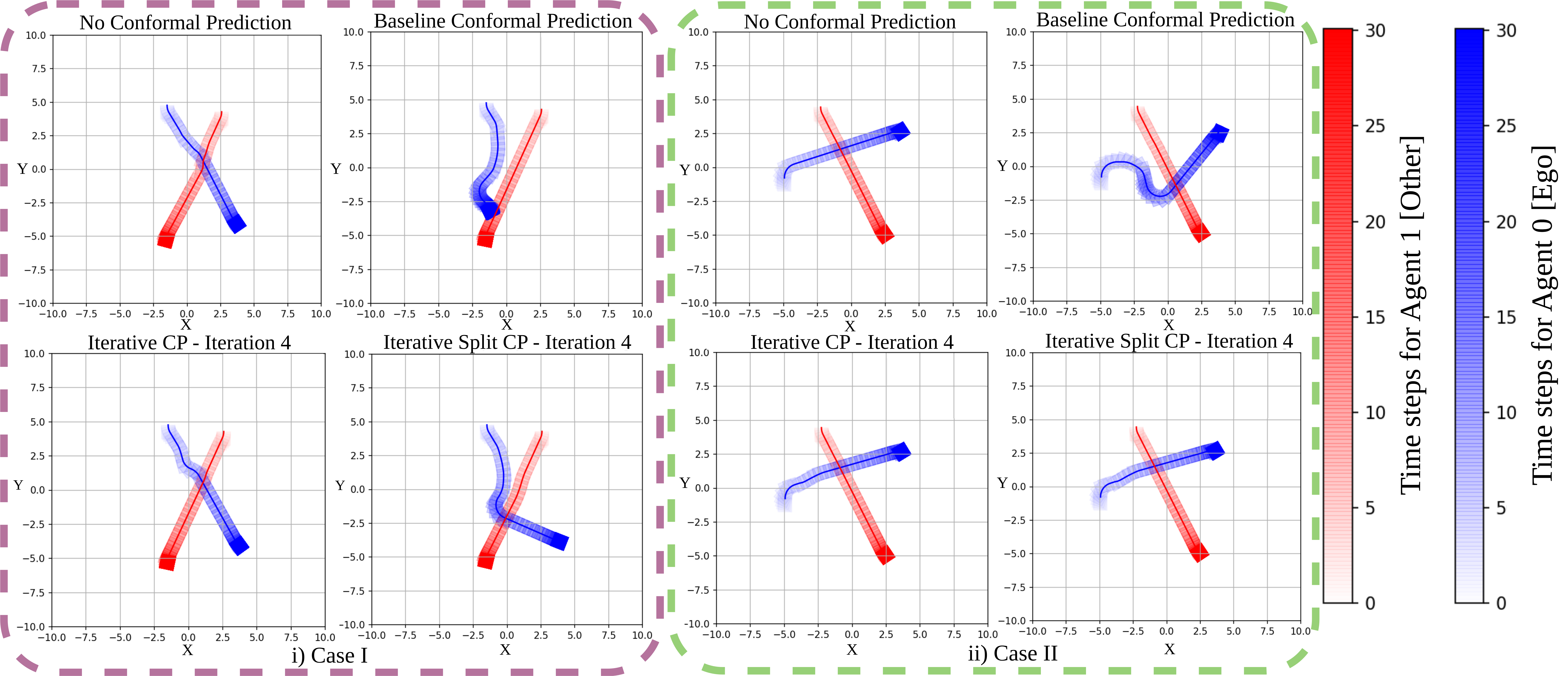}
    \caption{\textit{Comparison for $2$ agents across test cases}: Our iterative methods produce trajectories closely matching those without conformal prediction. In Case I, the baseline conformal prediction (BCP) method leads to a deadlock, while in Case II, BCP exhibits significant trajectory deviation.}
    \label{fig:case1}
    \vspace{-2mm}
\end{figure}

\begin{table*}
\centering
\caption{Performance metrics comparison for 2 Agents}
\label{tab:combined_results_2_agents}
\renewcommand{\arraystretch}{0.6}
\begin{tabular}{lccccc c}
\toprule
\textbf{Method} & \textbf{Collision} & 
\multicolumn{2}{c}{\textbf{Deviation from Ref. Path}} &
\textbf{Misdetection} & \textbf{Avg Nav} & \textbf{Success} \\
\cmidrule(lr){3-4}
 & \textbf{(in \%)} & \textbf{Ego Agent} & \textbf{Other Agent} & \textbf{(in \%)} & \textbf{Time(in s)} & \textbf{(in \%)} \\
\midrule
\basefirst & 43.00 & 0.12 & 0.12 & - & \textbf{6.10} & 57.00 \\
\basesec   & \textbf{0.00}  & 0.63 & 0.08 & 0.30 & 7.64 & 90.50 \\
\cp        & 0.50  & 0.32 & 0.08 & 12.35 & 6.67 & 96.50 \\
\scp       & 1.50  & \textbf{0.29} & 0.08 & \textbf{4.10} & 6.83 & \textbf{97.00} \\
\bottomrule
\vspace{-2mm}
\end{tabular}
\end{table*}
% We compare the following approaches with $200$ test trajectories:
% \vspace{-5pt}
% \setlength{\leftmargini}{15pt} % remove left margin for first-level lists
% \begin{enumerate}
%     \item \textbf{No Conformal Prediction (NCP):}\label{NCP} No conformal prediction is applied (\(\boldsymbol{q}^{(r)} = \boldsymbol{0}\)).
%     \vspace{-10pt}
%     \item \textbf{Baseline Conformal Prediction (BCP):}\label{BCP} A predictor trained on agent interactions is used to collect $1250$ calibration trajectories. Following \eqref{eq:CP_Prediction}, we set $\delta=0.13$ for $\epsilon=0.15$. 
%     \vspace{-10pt}
%     \item \textbf{Iterative Conformal Prediction (ICP, Ours):}\label{ICP} Using the same predictor as BCP, iterative calibration is performed over $5$ iterations with $250$ trajectories per iteration. We set $\delta=0.1$ for  $\epsilon=0.15$, $\gamma=0.8$ for the smoothing parameter and stopping criteria is $\phi=0.1$.
%         \vspace{-10pt}
%     \item \textbf{Iterative Split Conformal Prediction (ISCP, Ours):}\label{ISCP} The predictor is trained with an equal dataset size as BCP, and calibration is performed as in ICP over $4$ iterations. The same $\delta, \epsilon, \gamma$ and $\phi$ parameters as above are used.
% \end{enumerate}
We consider a two-agent scenario where the ego agent employs conformal prediction while the other agent uses a nominal MPC controller without uncertainty quantification, \emph{i.e.}, $\pi_{\text{MPC}}(\boldsymbol{0})$. For conformal prediction, following~\eqref{eq:CP_Prediction} the failure (and misdetection) probability $\epsilon=0.15$ with $\delta=0.01$. The smoothing parameter $\gamma=0.8$ for ICP and $\gamma = 0.9$ for ISCP, and the stopping criteria $\phi=0.1$. For both baseline conformal prediction (BCP) and ICP, the predictor is trained on $|D_{tune}^{\text{BCP}}|=1000$ trajectories and with $|D_{cal}^{\text{BCP}}|=1000$. The ICP method uses $K=250$ for each iteration. For ISCP, $K_{tune}=250, \, K = 250$ trajectories are used per iteration for retraining the predictor and calibration respectively. The ISCP method converges in $r=4$ iterations while the ICP method converges in $r=4$ iterations. Please refer to Appendix \ref{appendix: 2-agent} for further details.

We use $200$ test trajectories for obtaining the results reported in Table \ref{tab:combined_results_2_agents} (the evaluation metrics used are explained in Appendix \ref{Appendix: Evaluation_Metrics}). The results show NCP leads to most collisions and unsafe behavior, while BCP, which does not split the calibration budget into iterations and uses the entire budget in one iteration,  exhibits overly conservative behavior with larger ego-agent path deviations and reduced success rates due to deadlocks, as demonstrated in Fig.~\ref{fig:case1}. In contrast, our ICP method adaptively updates the agent behavior towards a stable configuration, resulting in reduced deviations from the reference trajectory while still satisfying the desired safety levels. Moreover, our ISCP method, which retrains the predictor alongside iterative calibration, yields higher success rates, smaller ego-agent path deviations, and also maintains misdetections within the target bounds.

\vspace{-2mm}
\subsection{Three Agent Simulations}

% We compare the following approaches with $200$ test trajectories:
% \vspace{-5pt}
% \setlength{\leftmargini}{15pt} % remove left margin for first-level lists
% \begin{enumerate}
%     \item \textbf{No Conformal Prediction (NCP):}\label{NCP} No conformal prediction is applied (\(\boldsymbol{q}^{(r)} = \boldsymbol{0}\)).
%     \vspace{-10pt}
%     \item \textbf{Baseline Conformal Prediction (BCP):}\label{BCP} A predictor trained on agent interactions is used to collect $1250$ calibration trajectories. Following \eqref{eq:CP_Prediction}, we set $\delta=0.081$ for $\epsilon=0.10$. 
%     \vspace{-10pt}
%     \item \textbf{Iterative Conformal Prediction (ICP, Ours):}\label{ICP} Using the same predictor as BCP, iterative calibration is performed over $5$ iterations with $250$ trajectories per iteration. We set $\delta=0.076$ for  $\epsilon=0.10$, $\gamma=0.9$ for the smoothing parameter and stopping criteria is.
%         \vspace{-10pt}
%     \item \textbf{Iterative Split Conformal Prediction (ISCP, Ours):}\label{ISCP} The predictor is trained with an equal dataset size as BCP, and calibration is performed as in ICP over  iterations. The same $\delta, \epsilon, \gamma$ and $\phi$ parameters as above are used.
% \end{enumerate}
We consider a three-agent scenario in which all agents, in contrast to the two-agent setting discussed in Section~\ref{2_agent_results}, incorporate CP-based uncertainty predictions within their MPC planners. We use the misdetection probability $\epsilon=0.10$ with \(\delta=0.01\), $\gamma=0.9$, and $\phi=0.2$. For BCP, $|D_{tune}^{\text{BCP}}|=1000$ and $|D_{cal}^{\text{BCP}}|=1000$. The ICP uses $K=250$ for each iteration for calibration and ISCP uses $K_{tune}=250$ and $K=250$. The ICP and ISCP methods converge in $r=4$ iterations. However, we note that if we run ICP for more iterations, the prediction sets start changing again, indicating that the  fixed-point for ICP is unstable. Please see Appendix~\ref{appendix: 3-agent} for more details and experiments.

The results, summarized in Table~\ref{tab:combined_results_3_agent}, demonstrate that BCP, now overly optimistic in its behavior, fails to meet the desired misdetection coverage. This behavior contrasts with the two-agent case. In both cases, the unmodeled distribution shift causes the baseline to be either too optimistic or too conservative. When we instead split the budget across multiple iterations, both ICP and ISCP show better performance by satisfying the desired collision guarantees. We see that much like the 2-agent case, ISCP gives us improved performance across most metrics.
% Consequently, the ISCP method achieves improved performance, satisfying coverage guarantees and exhibiting greater robustness to the effects of distribution shift.
\begin{table*}
\centering
\caption{Performance metrics comparison for 3 Agents}
\label{tab:combined_results_3_agent}
\renewcommand{\arraystretch}{0.6}
\begin{tabular}{lccccc c}
\toprule
\textbf{Method} & \textbf{Collision} & 
\multicolumn{2}{c}{\textbf{Deviation from Ref. Path}} &
\textbf{Misdetection} & \textbf{Avg Nav} & \textbf{Success} \\
\cmidrule(lr){3-4}
 & \textbf{(in \%)} & \textbf{Ego Agent} & \textbf{Other Agent} & \textbf{(in \%)} & \textbf{Time(in s)} & \textbf{(in \%)} \\
\midrule
\basefirst & 35.17 & 0.17 & 0.23 & - & \textbf{8.72} & 24.50 \\
\basesec   & 2.33  & 0.65 & 0.72 & 40.40 & 11.04 & 88.50 \\
\cp        & 1.00  & 1.04 & 1.25 & 11.03 & 13.29 & 87.50 \\
\scp       & \textbf{0.00}  &\textbf{ 0.49 }& \textbf{0.50} & \textbf{7.10} & 11.90 & \textbf{97.00} \\
\bottomrule
\end{tabular}
\vspace{-2mm}
\end{table*}
\vspace{-4mm}
\section{Conclusion and Future Work}
We introduce an iterative CP framework for multi-agent reactive systems under \distributionshift. By iteratively recalibrating prediction regions, the proposed method achieves desired coverage guarantees while adapting to the distribution shift. We note, however, that the use of smoothing filters for prediction sets requires some hyperparameter tuning.  Future research will focus on developing adaptive strategies to mitigate hyperparameter sensitivity and exploring alternative statistical techniques beyond CP to describe the \distributionshift~such as using the statistical properties inherent to certain generative learning modules. Finally, we seek to deploy the theory presented on physical robotic platforms for real-world applications.
\vspace{-2mm}
\acks{The authors would like to thank Debajyoti Chakrabarti for his valuable insights and Professor Saumik Bhattacharya for supporting the inclusion of this work in the author’s Master’s thesis.}
\bibliography{main}
\clearpage
\appendix
\section{Experimental Details}

\subsection{Additional Details on System Dynamics and Environment}\label{appendix:env_setup}
For all the experiments, we consider the bicycle model (\cite{2006_Vehicle_Model}):
\begin{equation}
\begin{bmatrix}
x_{t+1} \\
y_{t+1} \\
\theta_{t+1} \\
v_{t+1}
\end{bmatrix}
=
\begin{bmatrix}
x_t + \Delta v_t \cos(\theta_t) \\
y_t + \Delta v_t \sin(\theta_t) \\
\theta_t + \dfrac{\Delta v_t}{l} \tan(\phi_t) \\
v_t + \Delta a_t
\end{bmatrix}
\end{equation}
where $p_t\coloneq (x_t, y_t)$ denotes the position in 2-dimensional cartesian-plane, $\theta_t$ is the vehicle's orientation, $v_t$ is the velocity of the vehicle, $l\coloneq1$ is the vehicle's length and $\Delta\coloneq0.1s$ is the sampling time. The control inputs for this model are $\phi_t$, which denotes the steering angle, and $a_t$, which denotes the acceleration of the vehicle.

In all the cases, the objective for the $0^{th}$ agent is to reach a goal region while avoiding the other reactive agents using the constraint function:
\begin{equation}
    c(p_{t+\kappa|t}, \hat{Y}_\kappa) \coloneq  \| p_{t+\kappa|t} - \hat{Y_{\kappa,j}}\| - \epsilon \quad 
\end{equation}
Here, $\epsilon$ is the safe distance to be considered between the agents. While this formulation has been defined for the ego-agent, it is extended to all the other agents as well, with the constraint to avoid all the other agents, including the ego-agent. The center of the goal region is considered as $x_{g}$ with zero velocity at goal point, \emph{i.e.}, $v_g=0$. 

The MPC Cost Function $J(.)$, which we define for the experiments, is the standard goal-reaching MPC formulation as shown:
\begin{equation}
    J(x_{\tau},u_{\tau}) \coloneq (x_{\tau}-x_g)^TQ(x_{\tau}-x_g) + u_\tau^TRu_\tau+0.01*v_\tau 
\end{equation}
where, \(Q=diag([1, 1, 0.001, 0.1])\) and $R=diag([0.001,0.01])$. Here, $diag([.])$ is used to represent a diagonal matrix whose diagonal elements are provided. We also consider a terminal constraint ($Q_T=diag([5.0,5.0,0.01,5.5])$). To solve the MPC optimization problem, we use CasADi (\cite{Casadi}) with Ipopt non-linear programming solver.
\subsection{Evaluation Metrics}\label{Appendix: Evaluation_Metrics}
The evaluation metrics used in Table \ref{tab:combined_results_2_agents} and \ref{tab:combined_results_3_agent} are explained below:
\begin{enumerate}
    \item \textbf{Collision Rate:} Percentage of trajectories in which at least one agent experiences a collision at any time, measuring the safety of the conformal prediction framework.
    \item \textbf{Misdetection:} Percentage of prediction instances where the ground-truth position of the agent lies outside the CP-based uncertainty set, computed across all agents, prediction horizons, and test samples, assessing the statistical validity of the uncertainty quantification. This metric provides a measure of the extent of distributional shift that occurs when the misdetection rate computed from the test trajectories varies from the misdetection rate imposed by the calibration dataset. 
    \item \textbf{Success Rate:} The percentage of trajectories in which the ego agent successfully reaches its designated goal within the specified time limit of 30 seconds and in which no collision occurs throughout the entire trajectory. This metric captures the effectiveness of the conformal prediction framework to avoid collisions and simultaneously to ensure over-conservativeness of the CP-based prediction frameworks, leading to overly conservative prediction regions.
    \item \textbf{Average Navigation Time:} The mean time taken by the agents that use CP-based methods to reach their goal position for all the successful attempts. Unsuccessful attempts are not considered in this calculation.
    \item \textbf{Deviation from Reference Path:} The deviation of the agents' actual trajectories from their nominal reference paths. This metric quantifies the size of the CP-based uncertainty predictions by measuring the distance between the executed trajectories and the planned reference.
\end{enumerate}
\subsection{Additional Details for Two Agent Simulations}\label{appendix: 2-agent}
In the two-agent environment, initial positions and goal locations are randomly initialized on a circle of diameter $10$, with a bias towards diagonal opposition and random displacement, subject to a minimum distance constraint to prevent initial collisions. This simulation setup ensures diverse, non-overlapping start and end configurations for both agents.

For the Iterative Split Conformal Prediction (ISCP) framework, iteration 0 begins with agents ignoring other agents, leading to frequent collisions, and serves as a base case for predictor training. In the Table \ref{tab:Results_2_agents_ISCP}, we show how the performance metrics vary over iterations.  Table~\ref{tab:Results_2_agents_ISCP} summarizes performance metrics over ISCP iterations, showing adaptation of the CP-based uncertainty sets to improve success rates and reduce path deviation, while maintaining statistical guarantees. Table~\ref{tab:misdetections_2_agent_iscp} records misdetection rates for each prediction step across iterations.
\begin{table}[H]
\centering
\caption{Performance metrics over iterations of the ISCP framework}
\label{tab:Results_2_agents_ISCP}
\renewcommand{\arraystretch}{0.8}
\begin{tabular}{ccccccc}
\toprule
\textbf{Iteration} & \textbf{Collision} & 
\multicolumn{2}{c}{\textbf{Deviation from Ref. Path}} &
\textbf{Misdetection} & \textbf{Avg Nav} & \textbf{Success} \\
\cmidrule(lr){3-4}
$r$ & \textbf{(in \%)} & \textbf{Ego Agent} & \textbf{Other Agent} & \textbf{(in \%)} & \textbf{Time(in s)} & \textbf{(in \%)} \\
\midrule
0 & 0.00 & 0.00 & - & - & 5.80 & 20.00 \\ 
1 & 3.00 &	0.56 &	0.11 & 16.00 &	7.10 & 88.00 \\ 
2 & 1.00 &	0.45 &	0.09 &  5.90 &	7.20 & 93.50 \\ 
3 & 0.50 &	0.40 &	0.07 &  3.00 &	6.89 & 94.50 \\ 
4 & 1.50 &	0.29 &	0.08 &  4.10 &	6.83 & 97.00 \\ 
\bottomrule
\end{tabular}
\end{table}

\begin{table}[H]
\centering
\caption{Misdetection Rate for 2 Agents of ISCP over $H=10$ prediction steps}
\label{tab:misdetections_2_agent_iscp}
\renewcommand{\arraystretch}{1.2} % reasonable row height for readability
{\scriptsize % or \footnotesize
\begin{tabular}{lccccccccccc}
\toprule
\textbf{r} & \textbf{All Steps} & \textbf{Step 1} & \textbf{Step 2} &
\textbf{Step 3} & \textbf{Step 4} & \textbf{Step 5} & \textbf{Step 6} &
\textbf{Step 7} & \textbf{Step 8} & \textbf{Step 9} & \textbf{Step 10} \\
\midrule
1 & 16.00 & 14.50 &	15.50 &	14.50 &	13.50 & 15.50 &	17.50 &	16.50 &	17.50 &	17.50 &	17.50 \\
2 &  5.90 &  6.00 &	 5.50 &	 5.00 &	 5.00 &	 5.00 &	 6.00 &	 6.00 &	 6.50 &	 7.00 &	 7.00 \\
3 &  3.00 &  4.00 &	 1.50 &	 2.00 &	 2.50 &	 3.00 &	 3.50 &	 3.00 &	 3.50 &	 3.50 &	 3.50 \\
4 & 4.10 &  3.50 &	 4.00 &	 5.00 &	 4.50 &	 4.00 &	 4.00 &	 4.00 &	 4.00 &	 4.00 &	 4.00 \\
\bottomrule
\end{tabular}
}
\end{table}

For the Iterative Conformal Prediction (ICP) framework, agents at iteration 0 plan collision-free trajectories without CP-based uncertainty predictions (i.e., with $\boldsymbol{q}^{(r)}=\boldsymbol{0}$). Table~\ref{tab:Results_2_agents_ICP} presents corresponding performance metrics across ICP iterations, while Table~\ref{tab:misdetections_2_agent_icp} reports misdetections. Notably, the misdetection rate is generally higher for ICP compared to ISCP, as tighter uncertainty sets are calibrated from available predictor outputs.

\begin{table}[H]
\centering
\caption{Performance metrics over iterations of the ICP framework}
\label{tab:Results_2_agents_ICP}
\renewcommand{\arraystretch}{0.6}
\begin{tabular}{ccccccc}
\toprule
\textbf{Iteration} & \textbf{Collision} & 
\multicolumn{2}{c}{\textbf{Deviation from Ref. Path}} &
\textbf{Misdetection} & \textbf{Avg Nav} & \textbf{Success} \\
\cmidrule(lr){3-4}
$r$ & \textbf{(in \%)} & \textbf{Ego Agent} & \textbf{Other Agent} & \textbf{(in \%)} & \textbf{Time(in s)} & \textbf{(in \%)} \\
\midrule
0 & 42.50 &	0.13 & 0.12 & -     & 6.04 & 57.50   \\
1 &  0.00 &	0.59 & 0.09 &  0.20 & 7.56 & 92.00   \\ 
2 & 6.50 &	0.30 & 0.08 & 16.60 & 6.68 & 90.50   \\
3 & 2.50 &	0.36 & 0.08 & 12.25 & 6.78 & 93.50   \\
4 & 0.50 &	0.32 & 0.08 & 12.35 & 6.67 & 96.50   \\
\bottomrule
\end{tabular}
\end{table}

\begin{table}[H]
\centering
\caption{Misdetection Rate for 2 Agents of ICP over $H=10$ prediction steps}
\label{tab:misdetections_2_agent_icp}
\renewcommand{\arraystretch}{1.2} % reasonable row height for readability
{\scriptsize % or \footnotesize
\begin{tabular}{lccccccccccc}
\toprule
\textbf{r} & \textbf{All Steps} & \textbf{Step 1} & \textbf{Step 2} &
\textbf{Step 3} & \textbf{Step 4} & \textbf{Step 5} & \textbf{Step 6} &
\textbf{Step 7} & \textbf{Step 8} & \textbf{Step 9} & \textbf{Step 10} \\
\midrule
1 & 0.20 & 0.00 & 0.0 & 0.0 & 0.0 & 0.50 & 1.00 & 0.50 & 0.0 & 0.0 & 0.0 \\
2 & 16.60 & 14.50 & 17.50 & 17.50 & 17.50 & 18.00 & 16.50 & 16.00 & 16.00 & 16.00 & 16.50 \\
3 & 12.25 & 9.00 & 11.00 & 12.00 & 12.50 & 12.50 & 13.00 & 12.50 & 13.00 & 13.50 & 13.50 \\
4 & 12.35 & 9.50 & 12.00 & 13.00 & 13.00 & 12.00 & 12.50 & 12.50 & 13.00 & 13.00 & 13.00 \\
\bottomrule
\end{tabular}
}
\end{table}

Plots of CP-based uncertainty predictions $\boldsymbol{q}^{(r)}$ over the prediction horizon for ISCP and ICP are shown in Figure~\ref{fig:C_radii_2_agent}.

\begin{figure}[htbp]
  \centering
  \begin{minipage}[b]{0.48\textwidth}
    \centering
    \includegraphics[width=\textwidth]{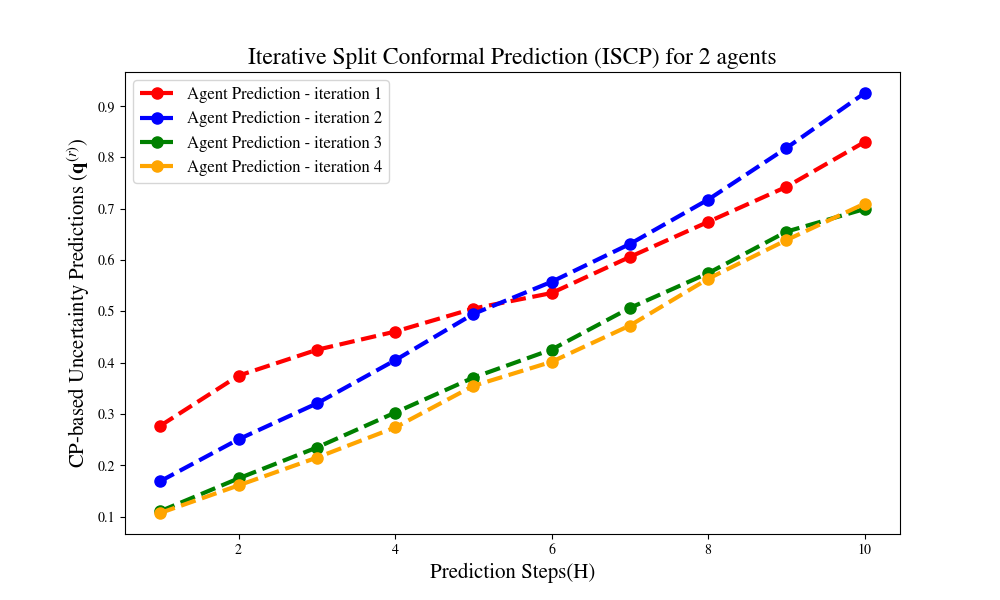}
    \captionof{figure}{ISCP: $\boldsymbol{q}^{(r)}$ over H}
    \label{fig:minipage1}
  \end{minipage}
  \hfill
  \begin{minipage}[b]{0.48\textwidth}
    \centering
    \includegraphics[width=\textwidth]{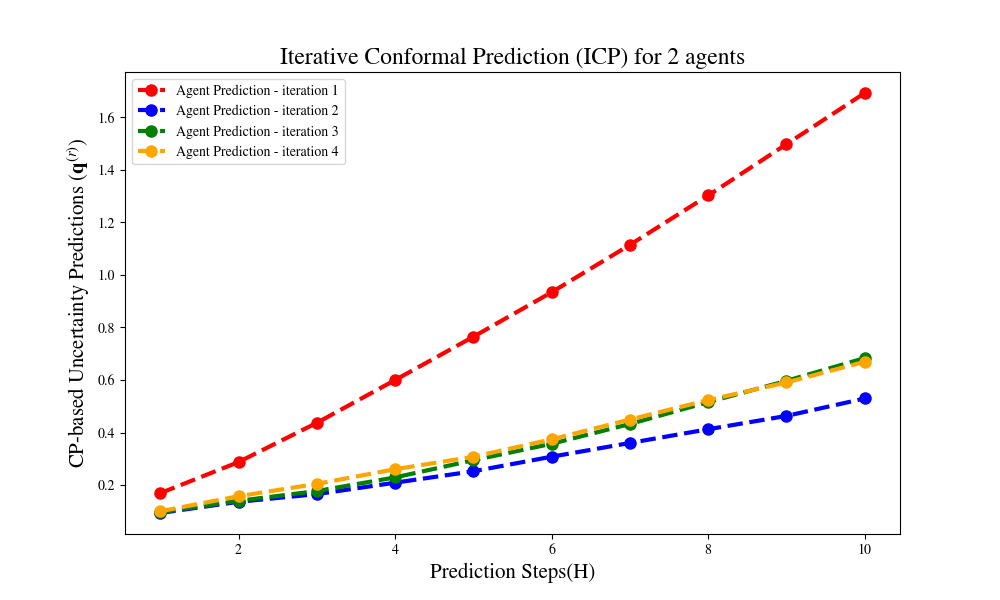}
    \captionof{figure}{ICP: $\boldsymbol{q}^{(r)}$ over H}
    \label{fig:minipage2}
  \end{minipage}
  \caption{CP-based uncertanity predictions for 2 agents}
  \label{fig:C_radii_2_agent}
\end{figure}

\subsection{Additional Details for Three Agent Simulations}\label{appendix: 3-agent}
The initial and goal points are obtained in a similar manner as explained in \ref{appendix: 2-agent}. We, however, consider the diameter of the circle as $15$ to accommodate more agents.

In the Table \ref{tab:Results_3_agents_ISCP}, we show how the performance metrics vary over iterations when all the agents deploy the MPC planner with CP-based uncertainty predictions. We observe that over iterations, the CP-based uncertainty sets update to increase the success rate and reduce deviations, while maintaining statistical guarantees. In Table \ref{tab:misdetections_2_agent_iscp}, we notice how the mis-detections (in \%) change over the iterations for each of the $H$ prediction steps.
\begin{table}[H]
\centering
\caption{Performance metrics over iterations of the ISCP framework}
\label{tab:Results_3_agents_ISCP}
\renewcommand{\arraystretch}{0.8}
\begin{tabular}{ccccccc}
\toprule
\textbf{Iteration} & \textbf{Collision} & 
\multicolumn{2}{c}{\textbf{Deviation from Ref. Path}} &
\textbf{Misdetection} & \textbf{Avg Nav} & \textbf{Success} \\
\cmidrule(lr){3-4}
$r$ & \textbf{(in \%)} & \textbf{Ego Agent} & \textbf{Other Agent} & \textbf{(in \%)} & \textbf{Time(in s)} & \textbf{(in \%)} \\
\midrule
0 &	58.17 &	0.00 &	0.00 &	    -    &	 7.31 &	19.50 \\  
1 &	 5.83 &	0.31 &	0.28 & 71.25 &	 9.72 &	88.00  \\
2 &	 0.33 &	0.64 &	0.52 & 26.80 &	13.06 &	91.00  \\
3 &	 0.00 &	0.72 &	0.57 & 12.25 &	12.30 &	96.50  \\
4 &	 0.00 &	0.49 &	0.50 & 7.10 &	11.90 &	97.00  \\
\bottomrule
\end{tabular}
\end{table}

\begin{table}[H]
\centering
\caption{Misdetection Rate for 3 Agents of ISCP over $H=10$ prediction steps}
\label{tab:misdetections_3_agent_iscp}
\renewcommand{\arraystretch}{1.2} % reasonable row height for readability
{\scriptsize % or \footnotesize
\begin{tabular}{lccccccccccc}
\toprule
\textbf{r} & \textbf{All Steps} & \textbf{Step 1} & \textbf{Step 2} &
\textbf{Step 3} & \textbf{Step 4} & \textbf{Step 5} & \textbf{Step 6} &
\textbf{Step 7} & \textbf{Step 8} & \textbf{Step 9} & \textbf{Step 10} \\
\midrule
1 & 71.25 & 32.83 & 40.17 & 50.67 & 62.17 & 75.50 & 82.00 & 88.00 & 92.83 & 94.00 & 94.33 \\
2 & 26.80 & 19.50 & 25.50 & 28.67 & 31.50 & 31.50 & 27.00 & 25.33 & 25.83 & 26.50 & 26.67 \\
3 & 12.25 & 8.50 & 13.67 & 13.67 & 14.50 & 15 & 12.83 & 11.67 & 10.83 & 10.67 & 11.17 \\
4 & 7.10 & 1.83 & 5.00 & 5.50 & 6.50 & 7.33 & 8.67 & 8.67 & 8.50 & 9.50 & 9.50 \\
\bottomrule
\end{tabular}
}
\end{table}

From the Figure \ref{fig:C_radii_3_agent_ISCP}, we observe the converging behavior for the CP-based uncertainty sets with the ISCP framework where $\gamma=0.9$. We observe that the predictor and the CP-based uncertainty sets converge to a more stable behavior for all the agents, with misdetection coverage guarantees satisfied.
\begin{figure}[htbp]
  \centering
  % First row with two figures
  \begin{minipage}[b]{0.325\textwidth}
    \centering
    \includegraphics[width=\linewidth]{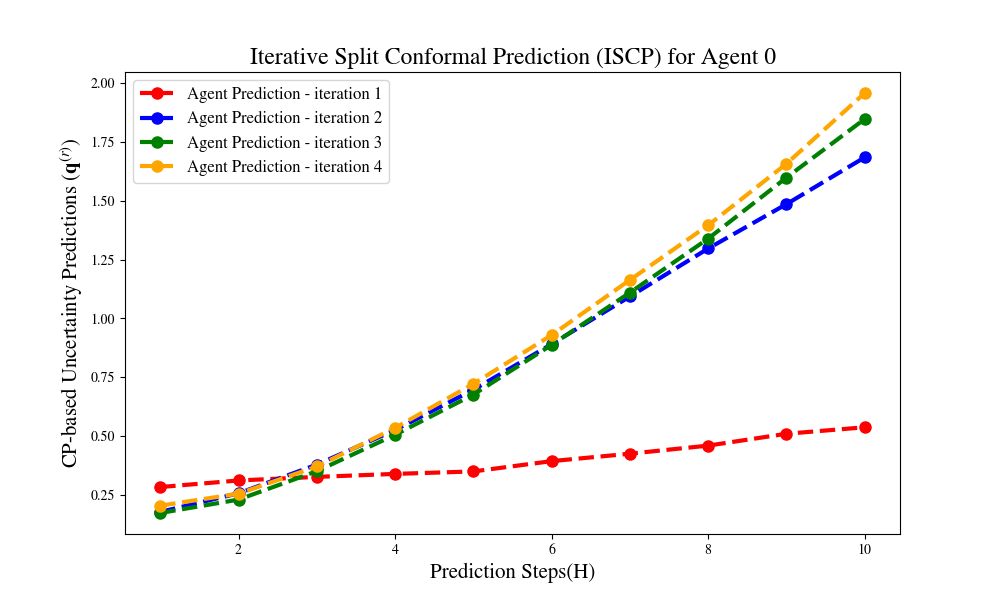}
    \captionof{figure}{Agent 0}
    \label{fig:minipage1}
  \end{minipage}
  % \hfill
  \begin{minipage}[b]{0.325\textwidth}
    \centering
    \includegraphics[width=\linewidth]{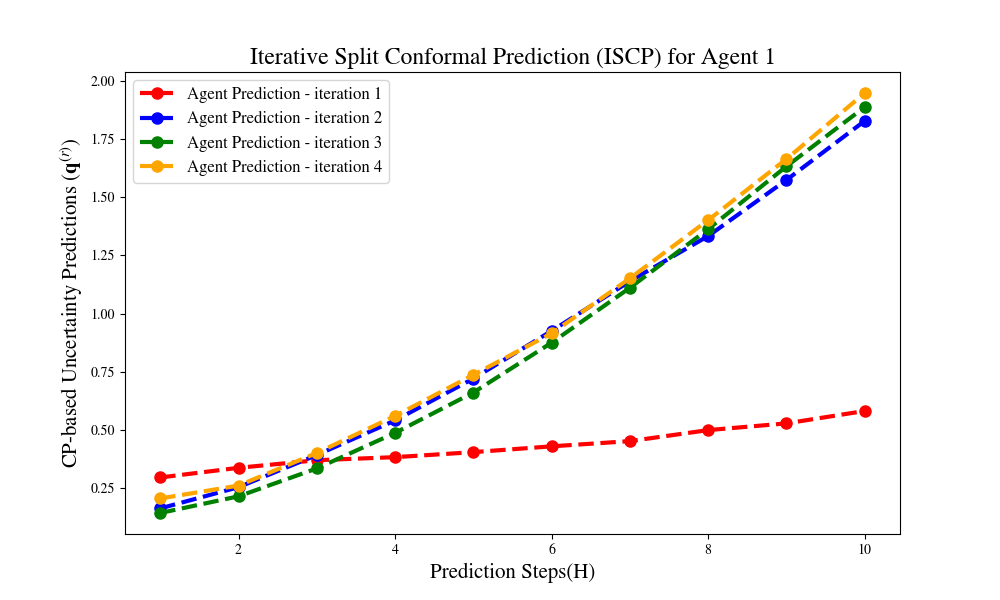}
    \captionof{figure}{Agent 1}
    \label{fig:minipage2}
  \end{minipage}
    % \hfill
  % Second row with one centered figure
  \begin{minipage}[b]{0.325\textwidth}
    \centering
    \includegraphics[width=\linewidth]{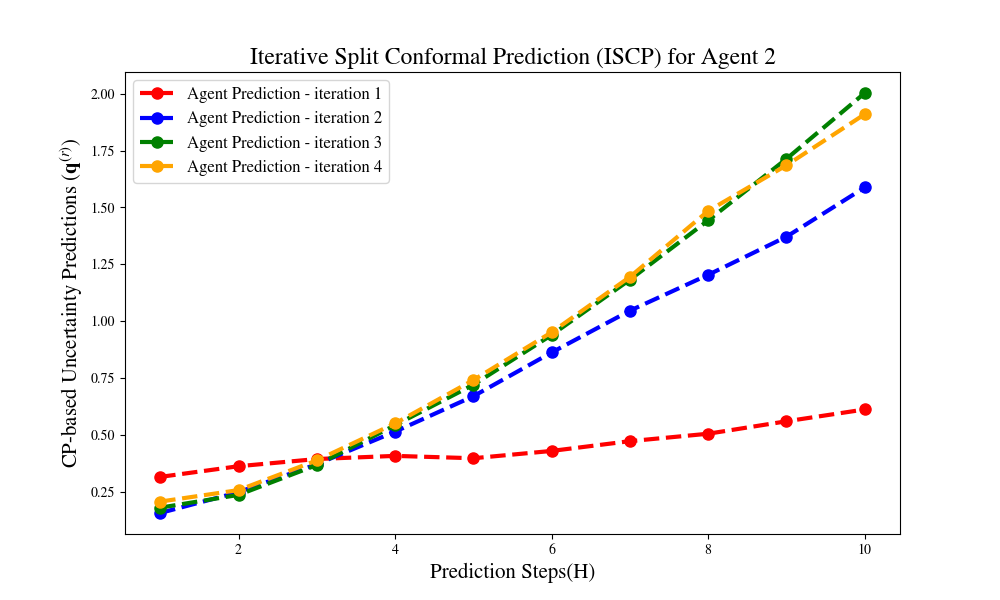}
    \captionof{figure}{Agent 2}
    \label{fig:minipage3}
  \end{minipage}
  
  \caption{ISCP: CP-based Uncertainty Predictions over Prediction Steps for 3 agents}
  \label{fig:C_radii_3_agent_ISCP}
\end{figure}

Table~\ref{tab:Results_3_agents_ICP} reports the performance metrics of the ICP framework over multiple iterations for a 3-agent environment. The CP-based uncertainty sets, \(\boldsymbol{q}^{(r)}\), adapt progressively to increase the success rate while aiming to satisfy the required misdetection coverage. Notably, iteration 2 achieves the targeted misdetection coverage with an improved success rate and reduced collision frequency. Subsequent iterations maintain comparable performance, though slight fluctuations in deviation and misdetection rates are observed.

\begin{table}[H]
\centering
\caption{Performance metrics over iterations of the ICP framework for 3 agents}
\label{tab:Results_3_agents_ICP}
\renewcommand{\arraystretch}{0.8}
\begin{tabular}{ccccccc}
\toprule
\textbf{Iteration} & \textbf{Collision} & 
\multicolumn{2}{c}{\textbf{Deviation from Ref. Path}} &
\textbf{Misdetection} & \textbf{Avg Nav} & \textbf{Success} \\
\cmidrule(lr){3-4}
$r$ & \textbf{(in \%)} & \textbf{Ego Agent} & \textbf{Other Agent} & \textbf{(in \%)} & \textbf{Time(in s)} & \textbf{(in \%)} \\
\midrule
0 &	32.83 &	0.16 & 0.21	&	    - &	 8.65 & 27.50 \\  
1 &	 2.00 &	0.81 & 0.76	&	40.90 &	13.37 & 87.50 \\  
2 &	 1.50 &	0.86 & 0.99	&	 9.78 &	12.94 & 89.50 \\  
3 &	 2.33 &	1.10 & 0.96	&	11.90 &	15.00 & 81.00 \\  
4 &	 1.00 &	1.04 & 1.25	&	11.03 &	13.29 & 87.50 \\ 
\bottomrule
\end{tabular}
\end{table}

\begin{table}[H]
\centering
\caption{Misdetection Rate for 3 Agents of ICP over $H=10$ prediction steps}
\label{tab:misdetections_2_agent_icp}
\renewcommand{\arraystretch}{1.2} % reasonable row height for readability
{\scriptsize % or \footnotesize
\begin{tabular}{lccccccccccc}
\toprule
\textbf{r} & \textbf{All Steps} & \textbf{Step 1} & \textbf{Step 2} &
\textbf{Step 3} & \textbf{Step 4} & \textbf{Step 5} & \textbf{Step 6} &
\textbf{Step 7} & \textbf{Step 8} & \textbf{Step 9} & \textbf{Step 10} \\
\midrule
1 & 40.90 &	34.67 &	34.50 &	36.00 &	37.17 &	40.33 &	43.67 &	44.17 &	45.83 &	46.17 &	46.50 \\
2 &  9.78 &	 9.17 &	 9.67 &	 9.00 &	 8.83 &	 9.00 &	10.00 &	10.83 &	10.50 &	10.83 &	10.00 \\
3 & 11.90 &	11.50 &	12.00 &	12.00 &	12.00 &	11.67 &	12.00 &	12.00 &	13.33 &	11.67 &	10.83 \\
4 & 11.03 &	12.50 &	12.67 &	12.00 &	11.50 &	11.17 &	11.00 &	10.83 &	 9.83 &	 9.50 &	 9.33 \\ 
\bottomrule
\end{tabular}
}
\end{table}

We investigate the effect of varying the smoothing parameter, considering $\gamma = {0.2, 0.8, 0.9}$. For each setting, the iterative procedure is executed for several steps beyond the predefined stopping criterion. The results indicate that the CP-based uncertainty predictions $\boldsymbol{q}^{(r)}$ approach a convergent value $\boldsymbol{q}^*$ but exhibit oscillatory behavior as iterations progress past the stopping condition. This phenomenon is particularly pronounced for larger values of $\gamma$, corresponding to lower smoothing. The convergence towards the stopping criterion typically occurs around iteration $4$; however, subsequent iterations lead to increased divergence in the uncertainty sets. As mentioned earlier, iteration 2 achieves the optimal balance, with minimal collisions and elevated success rates, as well as satisfying the targeted coverage. Notably, most of the \(\boldsymbol{q}^{(r)}\) are found to adapt to values near \(\boldsymbol{q}^{(2)}\) which might indicate a fixed-point.

\begin{figure}[htbp]
  \centering
  % First row with two figures
  \begin{minipage}[b]{0.325\textwidth}
    \centering
    \includegraphics[width=\linewidth]{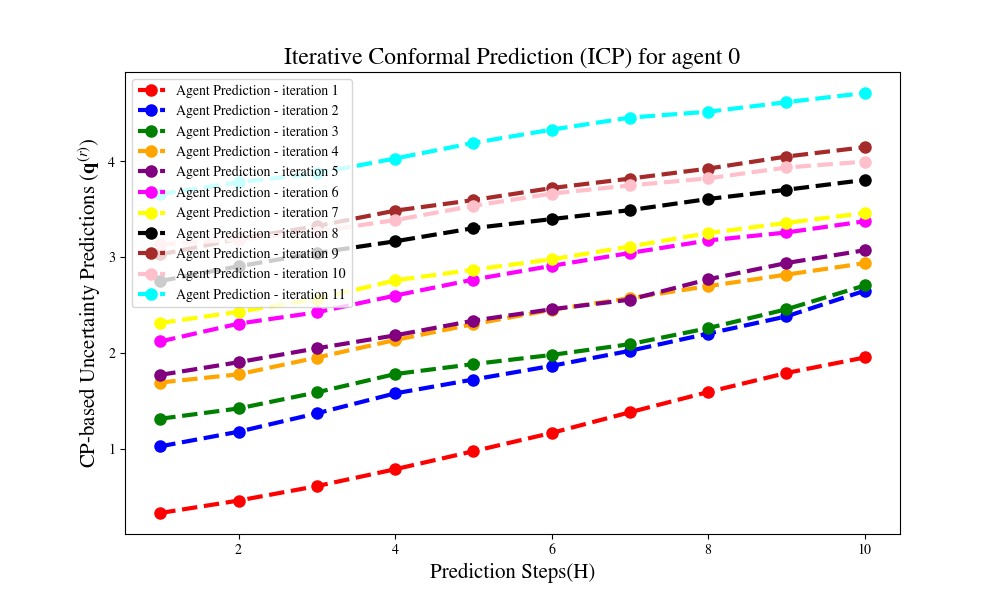}
    \captionof{figure}{Agent 0}
    \label{fig:minipage1}
  \end{minipage}
  % \hfill
  \begin{minipage}[b]{0.325\textwidth}
    \centering
    \includegraphics[width=\linewidth]{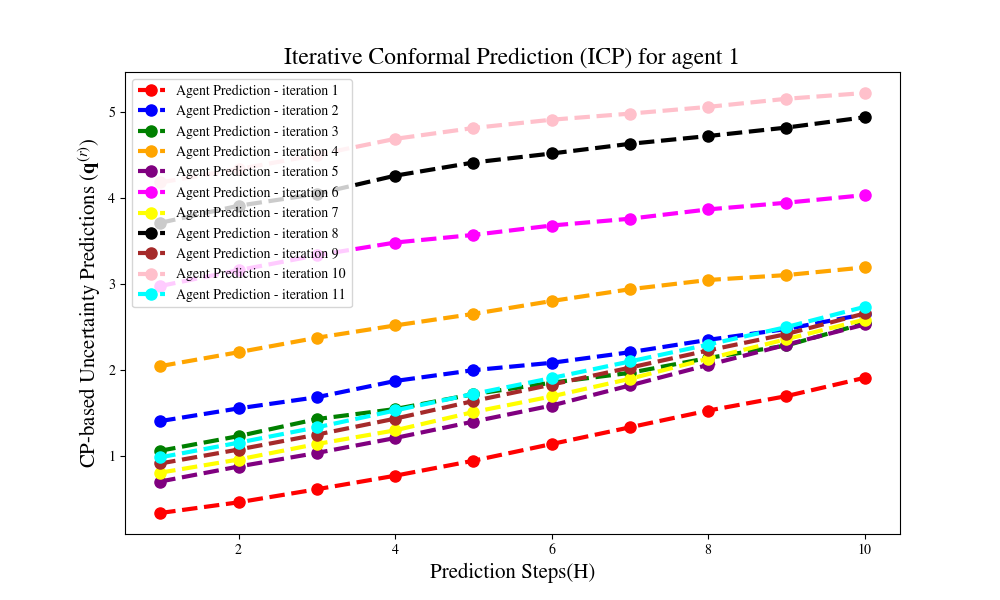}
    \captionof{figure}{Agent 1}
    \label{fig:minipage2}
  \end{minipage}
    % \hfill
  % Second row with one centered figure
  \begin{minipage}[b]{0.325\textwidth}
    \centering
    \includegraphics[width=\linewidth]{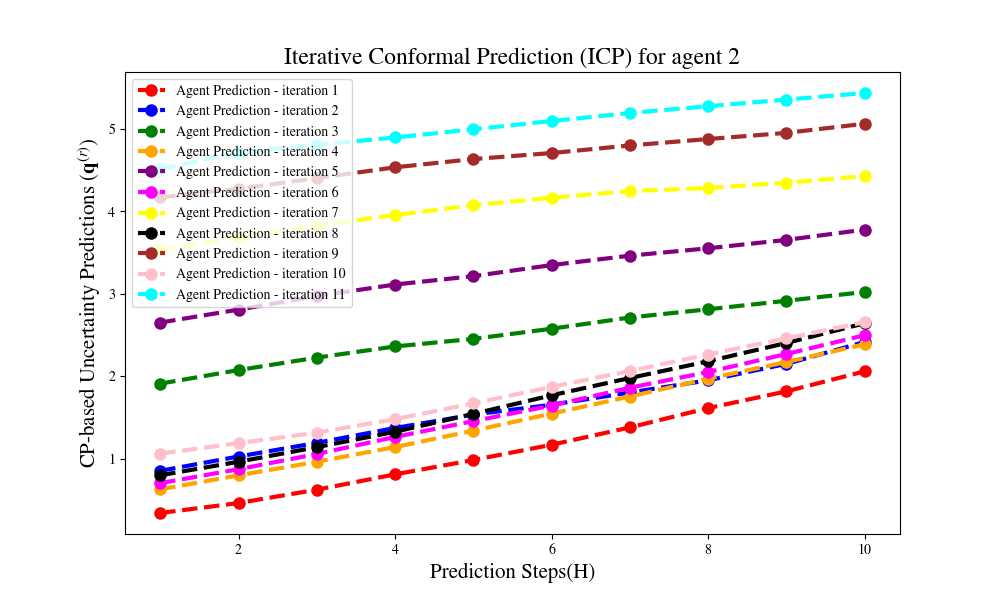}
    \captionof{figure}{Agent 2}
    \label{fig:minipage3}
  \end{minipage}
  
  \caption{ICP($\gamma=0.9$): CP-based Uncertainty Predictions over Prediction Steps for 3 agents}
  \label{fig:C_radii_3_agent_ISCP}
\end{figure}

\begin{figure}[htbp]
  \centering
  % First row with two figures
  \begin{minipage}[b]{0.325\textwidth}
    \centering
    \includegraphics[width=\linewidth]{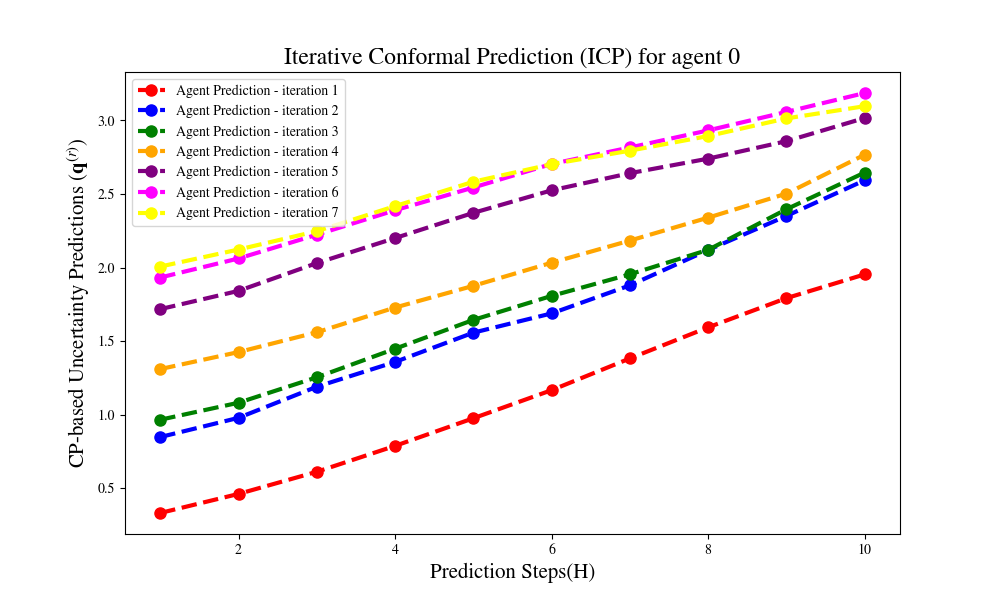}
    \captionof{figure}{Agent 0}
    \label{fig:minipage1}
  \end{minipage}
  % \hfill
  \begin{minipage}[b]{0.325\textwidth}
    \centering
    \includegraphics[width=\linewidth]{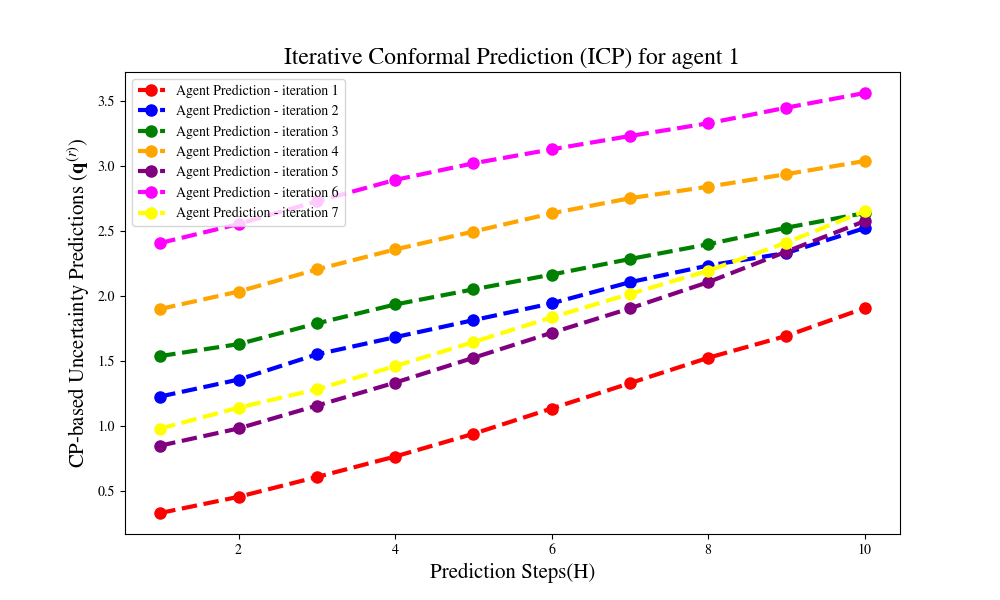}
    \captionof{figure}{Agent 1}
    \label{fig:minipage2}
  \end{minipage}
    % \hfill
  % Second row with one centered figure
  \begin{minipage}[b]{0.325\textwidth}
    \centering
    \includegraphics[width=\linewidth]{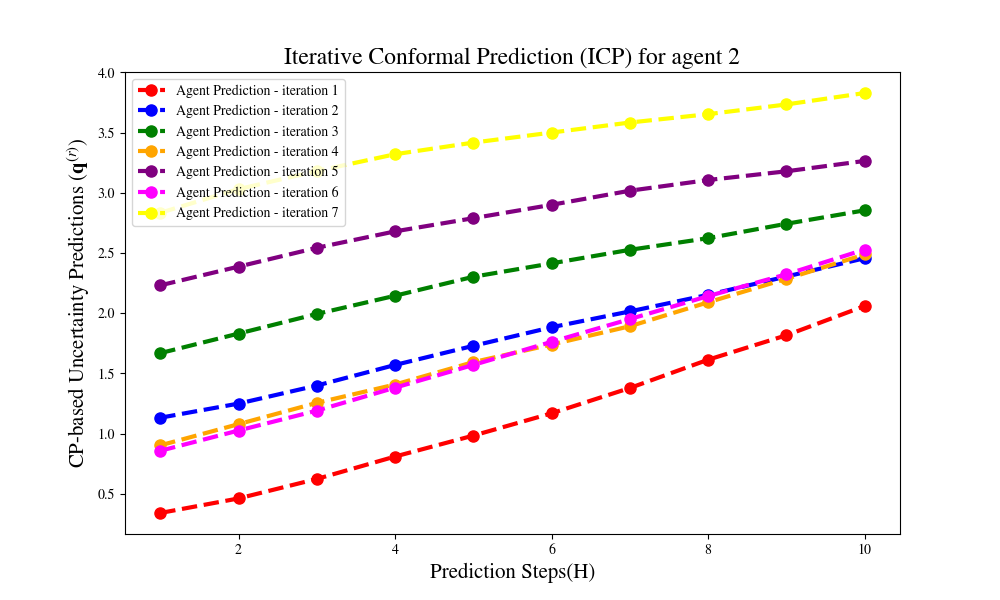}
    \captionof{figure}{Agent 2}
    \label{fig:minipage3}
  \end{minipage}
  
  \caption{ICP($\gamma=0.8$): CP-based Uncertainty Predictions over Prediction Steps for 3 agents}
  \label{fig:C_radii_3_agent_ISCP}
\end{figure}

\begin{figure}[htbp]
  \centering
  % First row with two figures
  \begin{minipage}[b]{0.325\textwidth}
    \centering
    \includegraphics[width=\linewidth]{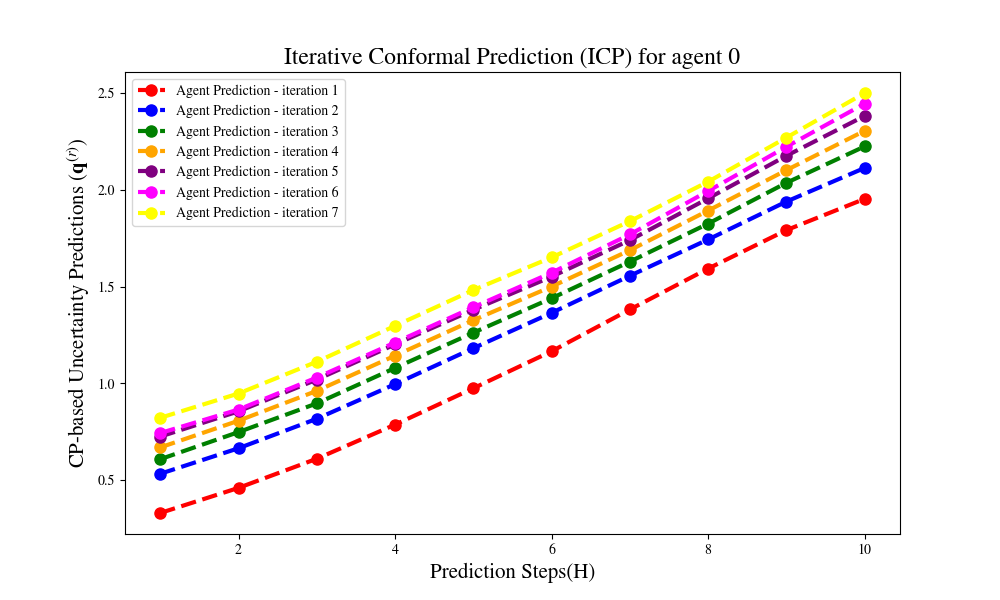}
    \captionof{figure}{Agent 0}
    \label{fig:minipage1}
  \end{minipage}
  % \hfill
  \begin{minipage}[b]{0.325\textwidth}
    \centering
    \includegraphics[width=\linewidth]{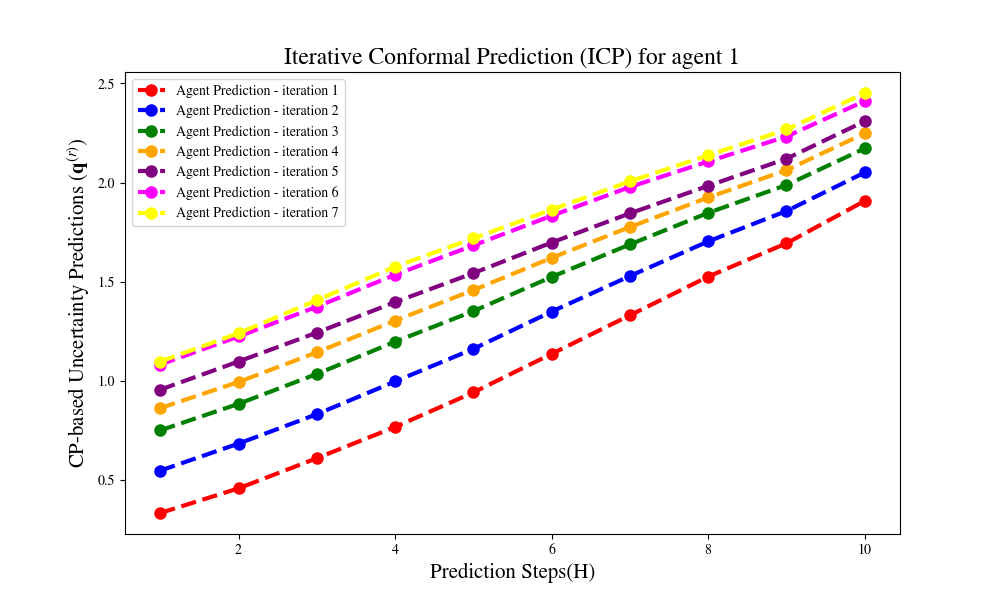}
    \captionof{figure}{Agent 1}
    \label{fig:minipage2}
  \end{minipage}
    % \hfill
  % Second row with one centered figure
  \begin{minipage}[b]{0.325\textwidth}
    \centering
    \includegraphics[width=\linewidth]{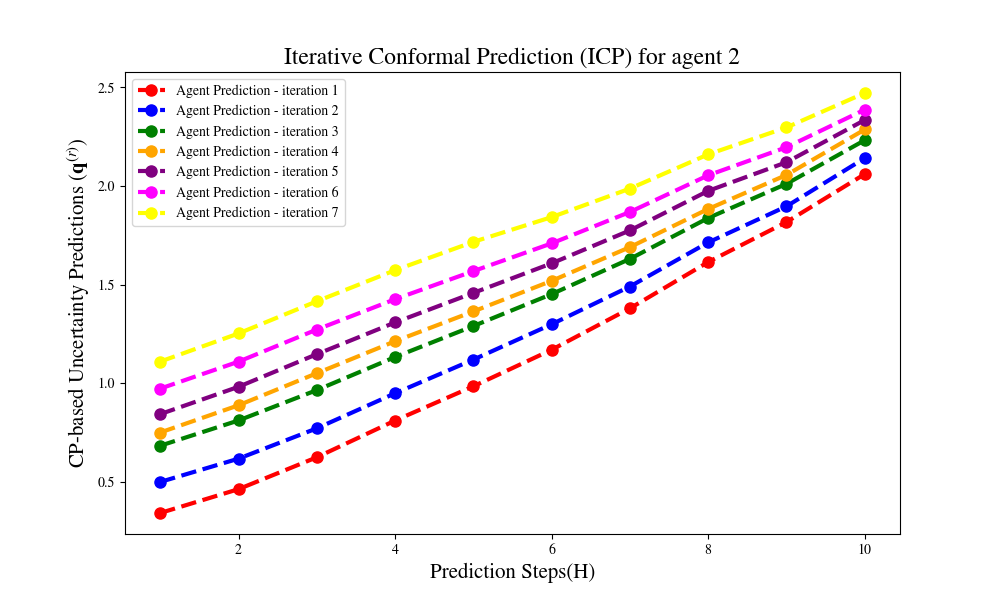}
    \captionof{figure}{Agent 2}
    \label{fig:minipage3}
  \end{minipage}
  
  \caption{ICP($\gamma=0.2$): CP-based Uncertainty Predictions over Prediction Steps for 3 agents}
  \label{fig:C_radii_3_agent_ISCP}
\end{figure}

\end{document}